\documentclass[conference]{IEEEtran} 
\IEEEoverridecommandlockouts
\usepackage{enumerate}
\usepackage{subfigure}
\usepackage{algorithm,amsthm}
\usepackage{algorithmic}
\usepackage{cite}
\usepackage{amsmath,amssymb,amsfonts}   
\usepackage{graphicx}
\usepackage{textcomp}
\usepackage{xcolor}
\usepackage{authblk}
\usepackage{float}
\usepackage[T1]{fontenc}
\usepackage[utf8]{inputenc}
\def\BibTeX{{\rm B\kern-.05em{\sc i\kern-.025em b}\kern-.08em
    T\kern-.1667em\lower.7ex\hbox{E}\kern-.125emX}}
    
\begin{document}

	\title{Utility Analysis and Enhancement of LDP Mechanisms in High-Dimensional Space}

\author{
	\vspace{0.08in}
	Jiawei Duan$^{\ast \ \dagger}$, Qingqing Ye $^\ast$, Haibo Hu $^\ast$\\
	$^\ast$ The Hong Kong Polytechnic University \\	
	$^\dagger$ Centre for Advances in Reliability and Safety\\
	\emph{ jiawei.duan@connect.polyu.hk; qqing.ye@polyu.edu.hk; haibo.hu@polyu.edu.hk }
}

\maketitle
	\begin{abstract}
Local differential privacy (LDP), which perturbs each user’s data locally and only sends the noisy version of her information to the aggregator, is a popular privacy-preserving data collection mechanism. In LDP, the data collector could obtain accurate statistics without access to original data, thus guaranteeing users' privacy. However, a primary drawback of LDP is its disappointing utility in high-dimensional space. Although various LDP schemes have been proposed to reduce perturbation, they share the same and naive aggregation mechanism at the collector's side. In this paper, we first bring forward an analytical framework to generally measure the utilities of LDP mechanisms in high-dimensional space, which can benchmark existing and future LDP mechanisms without conducting any experiment. Based on this, the framework further reveals that the naive aggregation is sub-optimal in high-dimensional space, and there is much room for improvement. Motivated by this, we present a re-calibration protocol $HDR4ME$ for high-dimensional mean estimation, which improves the utilities of existing LDP mechanisms without making any change to them. Both theoretical analysis and extensive experiments confirm the generality and effectiveness of our framework and protocol.
\end{abstract} 
		

	\newtheorem{definition}{\bf Definition}
	\newtheorem{theorem}{\bf Theorem}
	\newtheorem{lemma}{\bf Lemma}
	\newtheorem{corollary}{\bf Corollary}
	\section{Introduction}
In recent years, with growing number of IoT and smart devices, a huge amount of data becomes more accessible than ever \cite{8123562,8278156,ghayyur2018iot,liu2019dynapro,9154758}. Thanks to the advancement of modern machine learning and deep learning technologies, service providers and researchers nowadays can get insight into users' behavior and intention with simple clicks. However, together with the prevalence of these technologies emerge privacy concerns during the collection of sensitive data about users. To balance data utility and privacy disclosure, an effective and highly recognized solution is local differential privacy (LDP) \cite{raskhodnikova2008can,duchi2013local,ye2020local}, where the data collector only collects perturbed data from users.

Nevertheless, existing LDP mechanisms mostly focus on low-dimensional data, mainly because the statistics estimated in high-dimensional space have low accuracy. As users only authorize a limited privacy budget to the collector, the allocated privacy budget in each dimension is diluted as the number of dimensions increases, which leads to more information loss and poorer statistics accuracy. Although much attention has been paid to develop less perturbed LDP mechanisms for multi-dimensional data \cite{soria2013optimal,geng2015staircase,wang2019collecting,li2020estimating}, they are still not applicable in high-dimensional space. 

In this paper, we first propose an analytical framework that generalizes LDP mechanisms and derives their utilities in high-dimensional space, namely the probability density function of the deviation between the estimated mean and the true mean. The framework can serve as a benchmark to compare the utilities of various LDP mechanisms without conducting any experiment. Furthermore, our analysis shows the sub-optimality of the naive aggregation method of all LDP mechanisms --- the utility deterioration is attributed to the overwhelming noise caused by diluted privacy budget in high-dimensional space. As such, our second contribution in this paper is a one-off, non-iterative re-calibration protocol \emph{HDR4ME} (acronym for \underline{H}igh-\underline{D}imensional \underline{R}e-calibration for \underline{M}ean \underline{E}stimation). Through regularization and proximal gradient descent, this protocol re-calibrates the aggregated mean obtained from any LDP mechanism by suppressing the overwhelming noise and thus enhances its utility. Without any change on the LDP mechanism itself, \emph{HDR4ME} can be used as a general optimizer of existing LDP mechanisms in high-dimensional space. In summary, our main contributions are:
\begin{itemize}
\item {}We bring forward an analytical framework to measure the utilities of LDP mechanisms in high-dimensional space. This framework not only provides a theoretical baseline to benchmark existing and future LDP mechanisms, but also serves as a platform to compare their theoretical utilities in high-dimensional space.
\item {}We propose a re-calibration protocol \emph{HDR4ME} to enhance high-dimensional mean estimation and prove its superiority to the baseline. In particular, this protocol can be further extended to frequency estimation.
\item {} Based on both synthetic and real datasets, we conduct extensive experiments to validate our framework and evaluate our protocol for three state-of-the-art high-dimensional LDP mechanisms. Results show that the theoretical benchmark is consistent with the experimental results, and our protocol generally enhances the utilities.
\end{itemize}

The rest of this paper is organized as follows. In Section~\ref{sec::relatedwork}, we review the related literature. Section~\ref{sec::backgroud} introduces the fundamental concepts and formulates the problem. Then we introduce the analytical framework for high-dimensional LDP mechanisms in Section~\ref{sec::framework} and propose our mean estimation protocol in Section~\ref{sec::hdcsme}. Extensive experimental results are demonstrated in Section~\ref{sec::experiment}, and conclusions are made in Section~\ref{sec::conclusion}. 

	\section{Related Work}
\label{sec::relatedwork}
Dwork \emph{et al.} \cite{10.1007/11681878_14} formally present the definition of \emph{differential privacy} (DP) and propose the first DP mechanism, i.e., \emph{Laplace mechanism}. As for the local setting, Evfimievski \emph{et al.} \cite{evfimievski2003limiting} are among the first to introduce differential privacy at the side of individuals. Then Raskhodnikova \emph{et al.} \cite{raskhodnikova2008can} design a locally private mechanism $\gamma$-\emph{amplification randomizer}. Later on, Duchi \emph{et al.} \cite{duchi2013local} study the trade-off between local privacy budget and estimation utility, and derive bounds for \emph{local differential privacy} (LDP). LDP has been widely adopted in different domains, including itemset mining~\cite{wang2018locally}, marginal release~\cite{cormode2018marginal, zhang2018calm}, time series data release~\cite{ye2021beyond}, graph data analysis~\cite{sun2019analyzing, ye2020towards, ye2021lfgdpr}, key-value data collection~\cite{ye2019privkv, gu2020pckv,ye2021privkvm} and private learning~\cite{zheng2019bdpl, zheng2020protecting}. The most relevant problems to this paper include two aspects, namely, mean estimation by LDP and high-dimensional LDP. 

\subsection{Mean Estimation by LDP}
Dwork \emph{et al.} \cite{10.1007/11681878_14} initially propose \emph{Laplace mechanism} for mean estimation for centralized DP, which can also be applied to the local setting. Afterwards, several LDP frameworks, such as a variant of \emph{Laplace mechanism} referred to as \emph{SCDF} \cite{soria2013optimal} and \emph{Staircase mechanism} \cite{geng2015staircase}, perturb values with less noise. Note that the perturbed values of  these mechanisms range from negative to positive infinity, so they are classified as \emph{unbounded mechanisms} in this paper. On the contrary, \emph{bounded mechanisms} perturb values into a finite domain. Duchi \emph{et al.} \cite{duchi2018minimax} present one whose outputs are binary. To overcome the shortcoming of binary output, Wang \emph{et al.} \cite{wang2019collecting} propose \emph{Piecewise mechanism} and \emph{Hybird mechanism}. With continuous and bounded outputs, their utilities are improved. More recently, Li \emph{et al.} \cite{li2020estimating} propose \emph{square wave mechanism} where the perturbation is more centered than Piecewise, and the utility is therefore superior.


\subsection{High-Dimensional LDP}
The most critical challenge to adopt LDP in high-dimensional space is the utility degradation, a.k.a., the dimensionality curse. In general, there are two streams of methodology to cope with it. One is dimensionality reduction. As for non-local privacy data publication, Ren \emph{et al.} \cite{ren2018textsf} study frequency estimation based on \emph{Lasso Regression} and \emph{EM} algorithm. By \emph{principal components analysis} (PCA), Ge \emph{et al.} \cite{ge2018minimax} propose \emph{DPS-PCA} for interactive LDP while Wang \emph{et al.} \cite{wang2020principal} consider PCA for non-interactive LDP. Besides, Bassily \cite{bassily2019linear} studies linear queries estimation in high-dimensional LDP. The other methodology is correlation-based privacy budget allocation. Chatzikokolakis \emph{et al.} \cite{chatzikokolakis2013broadening} use metric $d_h$ to measure the similarity between two dimensions in DP. Larger $d_h$ indicates lower similarity, which requires more privacy budget in those dimensions. Alvim \emph{et al.} \cite{alvim2018local} extend this metric to LDP. Similarly, Li \emph{et al.} \cite{li2019information} calculate the respective information entropy of all dimensions while Du \emph{et al.} \cite{du2021collecting} use covariance of different dimensions to allocate privacy budget accordingly. 

It is noteworthy that almost all these works have limited the application scope in specific scenarios. Furthermore, many solutions have high computational cost at the user side \cite{duchi2018minimax,ge2018minimax,wang2020principal,chatzikokolakis2013broadening,alvim2018local}. This work, on the other hand, enhances high-dimensional LDP mean estimation by only involving the data collector. In addition, it is a general optimization that is irrespective of the LDP mechanisms.

	\section{Preliminaries and Problem Definition}
\label{sec::backgroud}
\subsection{Local Differential Privacy}
In LDP, let $n$ denote the number of users and tuple $\boldsymbol{t}_i\left(1\leq i \leq n\right)$ denote the $i$-th user’s private data. To ensure privacy, each tuple $\boldsymbol{t}_i$ is locally perturbed into $\boldsymbol{t}_i^*$ by a certain perturbation mechanism $\mathcal{M}$. Afterwards, only perturbed tuples $\left\{\boldsymbol{t}_i^*|1\leq i \leq n\right\}$ are sent to the data collector. Table \ref{tab::note} summarizes the notations used throughout this paper. Given privacy budget $\epsilon>0$ which indicates the privacy protection level, $\epsilon$-\emph{local differential privacy} is formally defined as follows:

\begin{table}
	\caption{Notations}
	\begin{center}
			\begin{tabular}{|c|c|}\hline
				\textbf{Symbol}&\textbf{Meaning} \\\hline
				$n$&number of users \\\hline
				$d$&number of dimensions\\\hline
				$\mathcal{M}$&perturbation mechanism\\\hline
				$\boldsymbol{t}_i$&user’s private tuple\\\hline
				$\boldsymbol{t}_i^{*}$&user’s perturbed tuple\\\hline
				$m$&number of sampled dimensions\\\hline
				$r$&aggregator's received reports\\\hline
				$\boldsymbol{\bar\theta}$&original mean \\\hline
				$\boldsymbol{\hat\theta}$&estimated mean \\\hline
				$\boldsymbol{\theta^*}$&enhanced mean \\\hline
				$\mathcal{L}$&loss function \\\hline
				$\mathcal{R}$&regularizer \\\hline
				$\boldsymbol{\lambda}^*$&regularization term \\\hline
		\end{tabular}
	\end{center}
	\label{tab::note}
\end{table}

\begin{definition}
 \emph{($\epsilon$-local differential privacy)} A randomized perturbation mechanism $\mathcal{M}$ satisfies $\epsilon$-local differential privacy if and only if for any pair of tuples $\boldsymbol{t}_i, \boldsymbol{t}_j$, the following inequality always holds:
\begin{equation}
\frac{\Pr\left(\mathcal{M}\left(\boldsymbol{t}_i\right)=\boldsymbol{t}^*\right)}{\Pr\left(\mathcal{M}\left(\boldsymbol{t}_j\right)=\boldsymbol{t}^*\right)}\leq\exp\left(\epsilon\right)
\end{equation}
\end{definition}
In essence, LDP guarantees that given prior knowledge $\boldsymbol{t}^{*}$, it is unlikely for the data collector to identify the data source with high confidence. Privacy budget $\epsilon$ controls the trade-off between privacy protection level and utility. Lower privacy budget means stricter privacy preservation and therefore poorer utility.

\subsection{Problem Definition}
\label{subsec::mean}
In high-dimensional settings, each $\boldsymbol{t}_i$ consists of $d$ numerical dimensions $\boldsymbol{t}_{i1},\boldsymbol{t}_{i2}, …, \boldsymbol{t}_{id}$. Without loss of generality, we focus on mean estimation throughout this paper and assume that the domain of any dimension ranges from $[-1,1]$. Unless otherwise specified, we respectively use $\mathbb{E}(\cdot)$ and $Var(\cdot)$ to denote the expectation and the variance of a random variable.

\textbf{Mean Estimation.} We follow a common and general approach for LDP mechanisms to support high-dimensional data \cite{wang2019collecting,nguyen2016collecting,duchi2018minimax,wang2017locally}. Given a total privacy budget $\epsilon$, each user randomly reports $m(1\leq m \leq d)$ dimensions of her perturbed data to the collector, with budget $\epsilon/m$ allocated to each dimension so that $\epsilon$-LDP still holds. Let $r_j$ denote the number of reports that the data collector receives in the $j$-th dimension, and obviously $\mathbb{E}(r_i)=\frac{nm}{d}$ because randomly reporting $m$ out of $d$ dimensions from $n$ users' data is statistically equal to reporting $d$ dimensions from $\frac{nm}{d}$ users. The data collector aggregates and estimates the mean of the $j$-th dimensions as $\boldsymbol{\hat\theta}_j=\frac{1}{r_j}\sum_{i=1}^{r_j} \boldsymbol{t}^*_{ij}$, so the estimated $d$-dimensional mean is  $\boldsymbol{\hat\theta}=(\boldsymbol{\hat\theta}_1,\boldsymbol{\hat\theta}_2,...,\boldsymbol{\hat\theta}_{d-1},\boldsymbol{\hat\theta}_d)^\intercal$. Note that the original mean of users is $\boldsymbol{\bar\theta}=\frac{1}{n}\sum_{i=1}^{n}\boldsymbol{t}_{i}$. 
Our objective is for the estimated mean $\boldsymbol{\hat\theta}$ to be as close to the original mean $\boldsymbol{\bar\theta}$ as possible. Therefore, we adopt the following utility metrics which can measure their difference.

\textbf{Utility Metrics.} Theoretically, their difference can be measured by the \emph{Euclidean distance}, i.e., 
\begin{align}
	\Vert\boldsymbol{\hat\theta}-\boldsymbol{\bar\theta}\Vert_2=\sqrt{\sum_{j=1}^d |\boldsymbol{\hat\theta}_j-\boldsymbol{\bar\theta}_j|^2}
	\label{equ::euc}
\end{align}  
Following \cite{10.1007/11681878_14,wang2019collecting,li2020estimating}, we adopt \emph{mean square error} (MSE) to measure experimental error, namely, the average squared difference between estimated means and original means over all dimensions, i.e.,  
\begin{align}	MSE(\boldsymbol{\hat\theta})=\frac{1}{d}\sum_{j=1}^d |\boldsymbol{\hat\theta}_j-\boldsymbol{\bar\theta}_j|^2
\label{equ::mse}
\end{align}
Applying Equation \ref{equ::euc} to Equation \ref{equ::mse}, we have  $MSE(\boldsymbol{\hat\theta})=\frac{1}{d}\Vert\boldsymbol{\hat\theta}-\boldsymbol{\bar\theta}\Vert^2_2$, which means that the theoretical analysis on $\Vert\boldsymbol{\hat\theta}-\boldsymbol{\bar\theta}\Vert_2$ can predict how MSE varies without conducting any experiment.

In the rest of this paper, we focus on two tasks for mean estimation. First, we analyze the utilities of various LDP mechanisms when they are extended to high-dimensional space (see Section~\ref{sec::framework}). Second, we design a re-calibration protocol to enhance the utility of any LDP mechanism in high-dimensional space without modifying it (see Section~\ref{sec::hdcsme}).

	\section{An Analytical Framework for High-Dimensional LDP}
\label{sec::framework}
In LDP literature, there are a lot of mechanisms that work in high-dimensional space. While a few are originally designed for such space~\cite{duchi2018minimax}, many others are extended or adapted to work in high-dimensional settings~\cite{10.1007/11681878_14,soria2013optimal,geng2015staircase,wang2019collecting,li2020estimating}. However, there is no theoretical work on benchmarking their utilities in high-dimensional space using a unified yardstick, based on which theoretical comparison and enhancement can be carried out. In this section, we provide such an analytical framework on {\bf mean estimation} to better understand the theoretical performance of existing works. Section \ref{subsec::repre} first reviews three state-of-the-art LDP mechanisms, based on which we present our framework in Section~\ref{subsec::framework}. Section \ref{subsec::case} provides a case study on benchmarking Piecewise mechanism \cite{wang2019collecting} and Square wave mechanism \cite{li2020estimating}, and Section \ref{subsec::error} provides the convergence rate of our framework. Without loss of generality, in what follows we assume each dimension has a normalized value domain $[-1,1]$, and some sampling techniques \cite{wang2019collecting,nguyen2016collecting,wang2017locally} are adopted so that each user only reports $m$ out of $d$ dimensions of her perturbed data to the data collector.

\subsection{Three State-of-the-Art High-Dimensional LDP Mechanisms}
\label{subsec::repre}
\noindent
\textbf{Laplace mechanism.} As a classic LDP mechanism, the advantage of \emph{Laplace mechanism} \cite{10.1007/11681878_14} is its simplicity. Given that a one-dimensional value $t_i$ in the range of $[-1,1]$, the perturbed value $t_i^{*}=t_i+Lap(2/\epsilon)$, where $Lap(\lambda)$ denotes a random variable that follows Laplace distribution with probability density function $f(x)=\frac{1}{2\lambda}\exp(-\frac{|x|}{\lambda})$. Note that the variance of $Lap(\lambda)$ is $2\lambda^2$ \cite{kotz2001laplace}. To extend it to high-dimensional values, each dimension is perturbed independently with a random variable $Lap(2m/\epsilon)$ to guarantee $\epsilon$-LDP. Since Laplace noise has zero mean, the data collector only needs to average all received tuples to achieve an unbiased mean estimation.

Generally, Laplace mechanism represents a class of LDP mechanisms \cite{10.1007/11681878_14,soria2013optimal,geng2015staircase} where the noise added to the original value ranges from negative to positive infinity. In our analytical framework, they are referred to as ``unbounded mechanisms''.

\noindent
\textbf{Piecewise mechanism.}
In one-dimensional Piecewise mechanism~\cite{wang2019collecting}, the perturbed value $t^*$ of an original value $t\in[-1,1]$ follows the distribution below:
\begin{equation}
\Pr(t^{*})=
\begin{cases}
\frac{e^{\epsilon}-e^{\epsilon/2}}{2e^{\epsilon/2}+2} & t^{*}\in\left[l(t),r(t)\right]\\
\frac{1-e^{-\epsilon/2}}{2e^{\epsilon/2}+2} &t^{*}\in\left[-Q,l(t))\cup(r(t),Q\right]
\end{cases},
\end{equation}
where
\begin{align*}
	Q&=\frac{e^{\epsilon}+e^{\epsilon/2}}{e^{\epsilon}-e^{\epsilon/2}} \\
	l(t)&=\frac{Q+1}{2} t-\frac{Q-1}{2}  \\	
	r(t)&=l(t)+Q-1
\end{align*}
In high-dimensional space, similar to Laplace mechanism, each reporting dimension independently carries out $\epsilon/m$-LDP. In contrast to Laplace mechanism, Piecewise mechanism perturbs the original value into a bounded domain $[-Q,Q]$, so such mechanisms are referred to as ``bounded mechanisms''.

\noindent
\textbf{Square wave mechanism.} This is yet another ``bounded'' LDP mechanism that improves Piecewise with more concentrated perturbation~\cite{li2020estimating}. In its one-dimensional form, for any original value $t \in [0,1]$, the perturbed value $t^*\in [-b,b+1]$ follows the distribution as below:
\begin{equation}
\Pr(t^{*})=
\begin{cases}
\frac{e^\epsilon}{2be^\epsilon+1} & \mbox{if } \left| t-t^*\right|<b\\
\frac{1}{2be^\epsilon+1} & \mbox{otherwise}
\end{cases},
\end{equation}
where $b=\frac{\epsilon e^\epsilon-e^\epsilon+1}{2e^\epsilon(e^\epsilon-1-\epsilon)}$. 
Similar to Piecewise mechanism, in high-dimensional space, each reporting dimension carries out $\epsilon/m$-LDP perturbation.

\subsection{A General Analytical Framework}
\label{subsec::framework}

In this subsection, we present our general framework for high-dimensional LDP mechanisms. As aforementioned, we first use a boolean $Bound$ to denote whether the perturbation of a certain LDP mechanism $\mathcal{M}$ has a finite ``boundary'' $B$. Then a $d$-dimensional LDP mechanism with privacy budget $\epsilon$ is generalized as follows:
\begin{enumerate}[1)]
	\item \emph{Perturbation:} Each user has a private tuple $\boldsymbol{t}_i \left(1 \leq i \leq n \right)$, among which $m$ dimensional values are perturbed and reported. For each dimension $j \left(1 \leq j \leq d \right)$, the mechanism obfuscates $\boldsymbol{t}_{ij}$ to $\boldsymbol{t}^*_{ij}$ with budget $\epsilon/m$. If $Bound(\mathcal{M})=1$, the perturbed tuple satisfies $\boldsymbol{t}^*_i=\mathcal{M}(\boldsymbol{t}_i)\in [-B,B]^d$, where $B$ is a both positive and finite value. Otherwise, the perturbed tuple satisfies  $\boldsymbol{t}^*_i=\mathcal{M}(\boldsymbol{t}_i)=\boldsymbol{t}_i+\boldsymbol{N}_i$, where $\boldsymbol{N}_i$ denotes a random tuple from $\mathbb{R}^d$.
	\item \emph{Calibration:} In each dimension $j$, the data collector receives $r_j$ reports, where $r=\mathbb{E}(r_j)=\frac{nm}{d}$. Letting $\boldsymbol{\delta}_{ij}$ denote the bias of $\mathbb{E}(\boldsymbol{t}^*_{ij})$, we have $\boldsymbol{\delta}_{ij}=\mathbb{E}(\boldsymbol{t}^*_{ij}-\boldsymbol{t}_{ij})$. Accordingly, the collector calibrates the perturbed values by $\boldsymbol{\delta}_{ij}$. Note that $\boldsymbol{\delta}_{ij}=0$ carries out unbiased estimation. 
	\item  \emph{Aggregation:} For mean estimation in $j$-th dimension, the mechanism averages all calibrated values to obtain the estimated mean $\boldsymbol{\hat{\theta}}_j=\frac{1}{r_j}\sum_{i=1}^{r_j} \boldsymbol{t}^*_{ij}$.
\end{enumerate}

Under this framework, we analyze the utility of high-dimensional LDP mechanisms based on the theoretical distance between the original mean $\boldsymbol{\bar\theta}$ and the estimated mean $\boldsymbol{\hat\theta}$ using \emph{Lindeberg–Lévy Central Limit Theorem} (CLT)\cite{SHANTHIKUMAR1984153,book}. Since each dimension is independently perturbed, we first model the deviation $\boldsymbol{\hat\theta}_j-\boldsymbol{\bar\theta}_j$ in one dimension. 
\begin{lemma}
For any $\mathcal{M}$ and $\epsilon/m$, $Var(\boldsymbol{t}^*_{ij})$ and $\boldsymbol{\delta}_{ij}$ are deterministic if $Bound(\mathcal{M})=0$ while correlated to $\boldsymbol{t}_{ij}$ if $Bound(\mathcal{M})=1$. 
\label{lem::var}
\end{lemma}
\begin{proof}
If $Bound(\mathcal{M})=0$, $Var(\boldsymbol{t}^*_{ij})=Var(\boldsymbol{t}_{ij}+\boldsymbol{N}_{ij})=Var(\boldsymbol{t}_{ij})+Var(\boldsymbol{N}_{ij})=Var(\boldsymbol{N}_{ij})$ while $\boldsymbol{\delta}_{ij}=\mathbb{E}(\boldsymbol{t}^*_{ij}-\boldsymbol{t}_{ij})=\mathbb{E}(\boldsymbol{N}_{ij})$. Since $\boldsymbol{N}_{ij}$ follows one perturbation, both $Var(\boldsymbol{t}^*_{ij})$ and $\boldsymbol{\delta}_{ij}$ are determinsitic. If $Bound(\mathcal{M})=1$, different $\boldsymbol{t}_{ij}$ correspond with different perturbations. Otherwise, $\boldsymbol{t}^*_{ij}$ would be totally independent from $\boldsymbol{t}_{ij}$. In this case, $Var(\boldsymbol{t}^*_{ij})$ and $\boldsymbol{\delta}_{ij}$ depend on $\boldsymbol{t}_{ij}$.
\end{proof}
Lemma \ref{lem::var} derives some common properties on $Var(\boldsymbol{t}^*_{ij})$ and $\boldsymbol{\delta}_{ij}$. Given $\mathcal{M}$ and $\epsilon/m$, $Var(\boldsymbol{t}^*_{ij})$ and $\boldsymbol{\delta}_{ij}$ are certain functions of $\epsilon/m$ if $Bound(\mathcal{M})=0$. Otherwise, they are certain functions of both $\epsilon/m$ and $\boldsymbol{t}_{ij}$. As long as the perturbation is known, we are able to provide $Var(\boldsymbol{t}^*_{ij})$ and $\boldsymbol{\delta}_{ij}$  considering $Var(\boldsymbol{t}^*_{ij})=\mathbb{E}({\boldsymbol{t}^*_{ij}}^2)-\mathbb{E}^2(\boldsymbol{t}^*_{ij})$ and $\boldsymbol{\delta}_{ij}=\mathbb{E}(\boldsymbol{t}^*_{ij}-\boldsymbol{t}_{ij})$. For further utility analysis, we assume that $Var(\boldsymbol{t}^*_{ij})$ and $\boldsymbol{\delta}_{ij}$ are already provided given certain $\mathcal{M}$ and $\epsilon/m$. Because $\lim\limits_{r_j\rightarrow\infty}\left(\frac{1}{r_j}\sum_{i=1}^{r_j}\boldsymbol{t}_{ij}-\frac{1}{n}\sum_{i=1}^{n}\boldsymbol{t}_{ij}\right)=0
$, the deviation $\boldsymbol{\hat\theta}_j-\boldsymbol{\bar\theta}_j$ can be simplified if $r_j\rightarrow\infty$:
\begin{equation}
\begin{aligned}
&\lim\limits_{r_j\rightarrow\infty}\boldsymbol{\hat\theta}_j-\boldsymbol{\bar\theta}_j=\lim\limits_{r_j\rightarrow\infty}\left(\frac{1}{r_j}\sum_{i=1}^{r_j}\boldsymbol{t}_{ij}^{*}-\frac{1}{n}\sum_{i=1}^{n}\boldsymbol{t}_{ij}\right)\\=&\lim\limits_{r_j\rightarrow\infty}\left(\frac{1}{r_j}\sum_{i=1}^{r_j}\boldsymbol{t}_{ij}^{*}-\frac{1}{r_j}\sum_{i=1}^{r_j}\boldsymbol{t}_{ij}+\frac{1}{r_j}\sum_{i=1}^{r_j}\boldsymbol{t}_{ij}-\frac{1}{n}\sum_{i=1}^{n}\boldsymbol{t}_{ij}\right)\\=&\lim\limits_{r_j\rightarrow\infty}\frac{1}{r_j}\sum_{i=1}^{r_j}\left(\boldsymbol{t}_{ij}^{*}-\boldsymbol{t}_{ij}\right)
\end{aligned}
\end{equation}
Suppose that $X$ is a random variable following standard normal distribution $X\sim \mathcal{N}(0,1)$, and its probability density function is $\phi (x)=\frac{1}{\sqrt{2\pi}} \exp(-\frac{x^2}{2})$, the following two lemmas establish the asymptotic distribution of the deviation in one dimension.
\begin{lemma}
$\lim\limits_{r_j\rightarrow\infty}\boldsymbol{\hat\theta}_j-\boldsymbol{\bar\theta}_j\sim\mathcal{N}\left(\mathbb{E}(\boldsymbol{N}_{ij}),\frac{Var(\boldsymbol{N}_{ij})}{r_j}\right)$, if $Bound(\mathcal{M})=0$.
\label{lem::normal0}
\end{lemma}
\begin{proof}
For $\forall j$, $\left\{\boldsymbol{t}^*_{ij}-\boldsymbol{t}_{ij}|1\leq i \leq r_j\right\}$ are independent and identically distributed (i.i.d.) random variables because $\boldsymbol{t}^*_{ij}-\boldsymbol{t}_{ij}=\boldsymbol{N}_{ij}$. According to \emph{Lindeberg–Lévy Central Limit Theorem} \cite{SHANTHIKUMAR1984153,book}, the following probability holds:
\begin{equation}
\begin{aligned}
&\lim\limits_{r_j\rightarrow\infty} \Pr\left(\frac{\frac{1}{r_j}\sum_{i=1}^{r_j} (\boldsymbol{t}^*_{ij}-\boldsymbol{t}_{ij})-\mathbb{E}(\boldsymbol{t}^*_{ij}-\boldsymbol{t}_{ij})}{\sqrt{Var(\boldsymbol{t}^*_{ij}-\boldsymbol{t}_{ij})/r_j}}\leq X\right)\\=&\lim\limits_{r\rightarrow\infty} \Pr\left(\frac{\boldsymbol{\hat\theta}_j-\boldsymbol{\bar\theta}_j-\mathbb{E}(\boldsymbol{N}_{ij})}{\sqrt{Var(\boldsymbol{N}_{ij})/r_j}}\leq X\right)\\=&\int_{-\infty}^{X} \phi(x)dx
\end{aligned}
\label{equ::normal0}
\end{equation}
Thus, $\lim\limits_{r_j\rightarrow\infty}\frac{\boldsymbol{\hat\theta}_j-\boldsymbol{\bar\theta}_j-\mathbb{E}(\boldsymbol{N}_{ij})}{\sqrt{Var(\boldsymbol{N}_{ij})/r_j}}$ follows standard normal distribution $\mathcal{N}(0,1)$, by which our claim is proven.
\end{proof}

In Lemma \ref{lem::var}, both $Var(\boldsymbol{t}^*_{ij})=Var(\boldsymbol{N}_{ij})$ and $\boldsymbol{\delta}_{ij}=\mathbb{E}(\boldsymbol{N}_{ij})$ are deterministic if $Bound(\mathcal{M})=0$. On this basis, we could approximate $\boldsymbol{\hat\theta}_j-\boldsymbol{\bar\theta}_j$ using a specific Gaussian distribution $\mathcal{N}(\mathbb{E}(\boldsymbol{N}_{ij}),Var(\boldsymbol{N}_{ij})/r_j)$ if $Bound(\mathcal{M})=0$. 

However, it is rather challenging if $Bound(\mathcal{M})=1$. Lemma \ref{lem::var} proves that different original values follow different perturbations if $Bound(\mathcal{M})=1$. Consequently, $\left\{\boldsymbol{t}^*_{ij}-\boldsymbol{t}_{ij}|1\leq i \leq r_j\right\}$ are probably not identically distributed, which does not satisfy the prerequisite of CLT \cite{SHANTHIKUMAR1984153,book}. 

Nevertheless, we are still able to use one Gaussian distribution to approximate the summation of elements in any particular subset $\left\{\boldsymbol{t}^*_{ij}-\boldsymbol{t}_{ij}\right\}$, where all original data have the same value, and therefore CLT can be applied. Note that $\left\{\boldsymbol{t}^*_{ij}-\boldsymbol{t}_{ij}|1\leq i \leq r_j\right\}$ can be divided into several particular subsets by different original values. Let $\left\{v_j|1\leq j\leq d\right\}$ denote numbers of different original values in each dimension, $\left\{{\boldsymbol{p}}_{zj}|\sum_{z=1}^{v_j} {\boldsymbol{p}}_{zj}=1\right\}$ denote their corresponding probabilities. As regards original data following continuous distribution, we discretize them with sampling. The following lemma establishes the asymptotic distribution of the deviation in one dimension if $Bound(\mathcal{M})=1$, where we assume $\left\{\boldsymbol{t}^*_{ij}|1\leq i\leq r_j, 1\leq j\leq d\right\}$ is in ascending order in each dimension.
\begin{lemma}
$\lim\limits_{r_j\rightarrow\infty}\boldsymbol{\hat\theta}_j-\boldsymbol{\bar\theta}_j\sim\mathcal{N}\left(\mathbb{E}(\boldsymbol{\delta}_{ij}),\frac{\mathbb{E}(Var(\boldsymbol{t}^*_{ij}))}{r_j}\right)$, where $\mathbb{E}(\boldsymbol{\delta}_{ij})=\sum_{z=1}^{v_j}\boldsymbol{p}_{zj} \boldsymbol{\delta}_{(\sum_{o=1}^{z} r_j{\boldsymbol{p}_{oj}})j}$ and $\mathbb{E}(Var(\boldsymbol{t}^*_{ij}))={\sum_{z=1}^{v_j}\boldsymbol{p}_{zj} Var\left(\boldsymbol{t}^*_{(\sum_{o=1}^{z} r_j{\boldsymbol{p}_{oj}})j}\right)}$, if $Bound(\mathcal{M})=1$.
\label{lem::normal1}
\end{lemma}

\begin{proof}
For $1\leq c\leq v_j$, the original data in $\left\{\boldsymbol{t}_{ij}|r_j\sum_{z=1}^{c-1} {\boldsymbol{p}}_{zj}< i \leq r_j\sum_{z=1}^{c} {\boldsymbol{p}}_{zj}\right\}$ share the same value. Therefore, $\left\{\boldsymbol{t}^*_{ij}-\boldsymbol{t}_{ij}|r_j\sum_{z=1}^{c-1} {\boldsymbol{p}}_{zj}< i \leq r_j\sum_{z=1}^{c} {\boldsymbol{p}}_{zj}\right\}$ are i.i.d. random variables. According to \emph{Lindeberg–Lévy Central Limit Theorem} \cite{SHANTHIKUMAR1984153,book}, the following probability holds if $r_j$ approaches $\infty$:
\begin{equation}
\begin{aligned}
& \Pr\left(\frac{\sum_{i=r_j\sum_{z=1}^{c-1} {\boldsymbol{p}}_{zj}+1}^{r_j\sum_{z=1}^{c} {\boldsymbol{p}}_{zj}} \left(\boldsymbol{t}^*_{ij}-\boldsymbol{t}_{ij}-\mathbb{E}(\boldsymbol{t}^*_{ij}-\boldsymbol{t}_{ij})\right)}{\sqrt{Var(\boldsymbol{t}^*_{ij}-\boldsymbol{t}_{ij}) r_j \boldsymbol{p}_{cj}}}\leq X\right)\\&= \Pr\left(\frac{\sum_{i=r_j\sum_{z=1}^{c-1} {\boldsymbol{p}}_{zj}+1}^{r_j\sum_{z=1}^{c} {\boldsymbol{p}}_{zj}} \left(\boldsymbol{t}^*_{ij}-\boldsymbol{t}_{ij}- \boldsymbol{\delta}_{ij}\right)}{\sqrt{Var(\boldsymbol{t}^*_{ij}) r_j \boldsymbol{p}_{cj}}}\leq X\right)\\&=\int_{-\infty}^{X} \phi(x)dx
\end{aligned}
\label{equ::normal1}
\end{equation}
Therefore, $\frac{\sum_{i=r_j\sum_{z=1}^{c-1} {\boldsymbol{p}}_{zj}+1}^{r_j\sum_{z=1}^{c} {\boldsymbol{p}}_{zj}} \left(\boldsymbol{t}^*_{ij}-\boldsymbol{t}_{ij}-\boldsymbol{\delta}_{ij}\right)}{\sqrt{Var(\boldsymbol{t}^*_{ij}) r_j \boldsymbol{p}_{cj}}}$ approximately follows standard normal distribution $\mathcal{N}(0,1)$. Next, we use \emph{Mathematical Induction} to complete the proof.

For $c=1$, Equation \ref{equ::normal1} establishes $\sum_{i=1}^{r_j \boldsymbol{p}_{1j}}(\boldsymbol{t}^*_{ij}-\boldsymbol{t}_{ij}-\boldsymbol{\delta}_{ij})\sim \mathcal{N}(0, {r_j \boldsymbol{p}_{1j}Var(\boldsymbol{t}^*_{(r_j \boldsymbol{p}_{1j})j})})$.

Suppose that $\sum_{i=1}^{r_j\sum_{z=1}^{c} {\boldsymbol{p}}_{zj}} (\boldsymbol{t}^*_{ij}-\boldsymbol{t}_{ij}-\boldsymbol{\delta}_{ij})\sim \mathcal{N}(0, {\sum_{z=1}^{c}r_j\boldsymbol{p}_{zj}Var(\boldsymbol{t}^*_{(\sum_{o=1}^{z} r_j{\boldsymbol{p}_{oj}})j})})$ holds for $1<c<v_j$, we have the following for $r_j\rightarrow\infty$:
\begin{equation}
\begin{aligned}
\sum_{i=1}^{r_j\sum_{z=1}^{c+1} {\boldsymbol{p}}_{zj}} (\boldsymbol{t}^*_{ij}-\boldsymbol{t}_{ij}-\boldsymbol{\delta}_{ij})=&\sum_{i=1}^{r_j\sum_{z=1}^{c} {\boldsymbol{p}}_{zj}} (\boldsymbol{t}^*_{ij}-\boldsymbol{t}_{ij}-\boldsymbol{\delta}_{ij})\\+& \sum_{i=r_j\sum_{z=1}^{c} {\boldsymbol{p}}_{zj}+1}^{r_j\sum_{z=1}^{c+1} {\boldsymbol{p}}_{zj}} (\boldsymbol{t}^*_{ij}-\boldsymbol{t}_{ij}-\boldsymbol{\delta}_{ij})
\end{aligned}
\end{equation}

According to Equation \ref{equ::normal1}, $\sum_{i=r_j\sum_{z=1}^{c} {\boldsymbol{p}}_{zj}+1}^{r_j\sum_{z=1}^{c+1} {\boldsymbol{p}}_{zj}} (\boldsymbol{t}^*_{ij}-\boldsymbol{t}_{ij}-\boldsymbol{\delta}_{ij})
$ follows $\mathcal{N}\sim (0, {r_j\boldsymbol{p}_{(c+1)j}Var(\boldsymbol{t}^*_{ij})})$. Given $Y\sim \mathcal{N}(\mu_1,\sigma^2_1)$ and $Z\sim \mathcal{N}(\mu_2,\sigma^2_2)$, $Y+Z \sim \mathcal{N}(\mu_1+\mu_2,{\sigma_1^2+\sigma_2^2})$ \cite{lemons2003introduction}. Therefore, we prove that for $1< c< v_{j}$ and $r_j\rightarrow\infty$,
\begin{equation}
\begin{aligned}
&\sum_{i=1}^{r_j\sum_{z=1}^{c+1} {\boldsymbol{p}}_{zj}} (\boldsymbol{t}^*_{ij}-\boldsymbol{t}_{ij}-\boldsymbol{\delta}_{ij}) \\&\sim \mathcal{N}\left(0,\sum_{z=1}^{c+1}r_j \boldsymbol{p}_{zj}Var(\boldsymbol{t}^*_{(\sum_{o=1}^{z} r{\boldsymbol{p}_{oj}})j})\right)
\end{aligned}
\end{equation}
Letting $c=v_{j}-1$, if $r_j$ approaches $\infty$, we have:
\begin{equation}
\begin{aligned}
&\Pr\left(\frac{\boldsymbol{\hat\theta}_j-\boldsymbol{\bar\theta}_j- \mathbb{E}(\boldsymbol{\delta}_{ij})}{\sqrt{\sum_{z=1}^{v_j}\boldsymbol{p}_{zj}Var(\boldsymbol{t}^*_{(\sum_{o=1}^{z} r_j{\boldsymbol{p}_{oj}})j})/ r_j}}\leq X\right)\\=& \Pr\left(\frac{\sum_{i=1}^{r_j}(\boldsymbol{t}^*_{ij}-\boldsymbol{t}_{ij})-\sum_{z=1}^{v_j}r_j\boldsymbol{p}_{zj} \boldsymbol{\delta}_{(\sum_{o=1}^{z} r_j{\boldsymbol{p}_{oj}})j}}{\sqrt{\sum_{z=1}^{v_j}r_j\boldsymbol{p}_{zj} Var(\boldsymbol{t}^*_{(\sum_{o=1}^{z} r_j{\boldsymbol{p}_{oj}})j})}}\leq X\right)\\=& \Pr\left(\frac{\sum_{i=1}^{r_j}(\boldsymbol{t}^*_{ij}-\boldsymbol{t}_{ij}-\boldsymbol{\delta}_{ij})}{\sqrt{\sum_{z=1}^{v_j}r_j\boldsymbol{p}_{zj} Var(\boldsymbol{t}^*_{(\sum_{o=1}^{z} r_j{\boldsymbol{p}_{oj}})j})}}\leq X\right)\\=&\int_{-\infty}^{X} \phi(x)dx
\end{aligned}
\end{equation}
by which our claim is proven.
\end{proof}

Note that $\mathbb{E}(\boldsymbol{\delta}_{ij})$ and $\mathbb{E}(Var(\boldsymbol{t}^*_{ij}))$ computes the expectations of $\boldsymbol{\delta}_{ij}$ and $Var(\boldsymbol{t}^*_{ij})$ in terms of $\boldsymbol{t}^*_{ij}$. In general, Lemma \ref{lem::normal0} and Lemma \ref{lem::normal1} establish that no matter how the original data is distributed, $\boldsymbol{\hat\theta}_j-\boldsymbol{\bar\theta}_j$ always approximates a normal distribution. However,  its variance is split into two cases. If $Bound(\mathcal{M})=0$, it is only decided by the distribution of perturbation; otherwise, it is collectively decided by distributions of both perturbation and original data. As such, given a certain dataset and a budget, we can model how $\boldsymbol{\hat\theta}_j-\boldsymbol{\bar\theta}_j$ varies in terms of any mechanism.

What if multiple or even high dimensions? Note that each dimension is independently perturbed with privacy budget $\epsilon/m$. As each dimension of the deviation approximates a one-dimensional normal distribution, we can model the deviation $\boldsymbol{\hat\theta}-\boldsymbol{\bar\theta}$ with one multivariate normal distribution. Following Lemma \ref{lem::normal0} or Lemma \ref{lem::normal1}, for $1\leq j\leq d$, $\boldsymbol{\hat\theta}_j-\boldsymbol{\bar\theta}_j$ approximates a normal distribution whose probability density function is $f(\boldsymbol{\hat\theta}_j-\boldsymbol{\bar\theta}_j)=\frac{1}{\sqrt{2\pi}\boldsymbol{\sigma}_j}\exp(-\frac{\boldsymbol{(\hat\theta}_j-\boldsymbol{\bar\theta}_j-\boldsymbol{\delta}_j)^2}{2\boldsymbol{\sigma}_j^2})$.
Then the following theorem models the deviation in high-dimensional space.
\begin{theorem}
	For any high-dimensional LDP mechanism, the probability density function (pdf) of  $\boldsymbol{\hat\theta}-\boldsymbol{\bar\theta}$ is:	
\begin{equation}
	\lim\limits_{r\rightarrow\infty}f(\boldsymbol{\hat\theta}-\boldsymbol{\bar\theta})=\frac{1}{(\sqrt{2\pi})^d\prod_{j=1}^d \boldsymbol{\sigma}_j} \exp\left(- \sum_{j=1}^d\frac{(\boldsymbol{\hat\theta}_j-\boldsymbol{\bar\theta}_j-\boldsymbol{\delta}_j)^2}{2\boldsymbol{\sigma}_j^2}\right)
\label{equ::pdf}
\end{equation}
\label{the::pdf}
\end{theorem}
\begin{proof}
Since each dimension is perturbed independently, we have:
\begin{equation}
\begin{aligned}
 f(\boldsymbol{\hat\theta}-\boldsymbol{\bar\theta})&=\prod_{j=1}^d f(\boldsymbol{\hat\theta}_j-\boldsymbol{\bar\theta}_j)=\prod_{j=1}^d \frac{1}{\sqrt{2\pi}\sigma_j}\exp(-\frac{\boldsymbol{(\hat\theta}_j-\boldsymbol{\bar\theta}_j-\boldsymbol{\delta}_j)^2}{2\sigma_j^2})\\&=\frac{1}{(\sqrt{2\pi})^d\prod_{j=1}^d \sigma_j} \exp(- \sum_{j=1}^d\frac{(\boldsymbol{\hat\theta}_j-\boldsymbol{\bar\theta}_j-\boldsymbol{\delta}_j)^2}{2\sigma_j^2})
 \end{aligned}
 \end{equation}
\end{proof}
 
As this pdf models how $\boldsymbol{\hat\theta}-\boldsymbol{\bar\theta}$ varies in high-dimensional space, we can accommodate almost all utility metrics for comparisons, including the supremum of the deviation. To benchmark different LDP mechanisms, intuitively the smallest supremum of the deviation $\sup\Vert\boldsymbol{\hat\theta}-\boldsymbol{\bar\theta}\Vert_2$ should have the best utility. However, due to the randomness in LDP mechanisms, the absolute supremum can be infinity. As such, the data collector can manually specify the supremum of deviation she wants to tolerate, and then calculate the corresponding probability for that supremum to hold using this pdf. Let $\boldsymbol{\xi}=\left(\boldsymbol{\xi}_1, ..., \boldsymbol{\xi}_d\right)^\intercal=\left(\sup\left|\boldsymbol{\hat\theta}_1-\boldsymbol{\bar\theta}_1\right|, ..., \sup\left|\boldsymbol{\hat\theta}_d-\boldsymbol{\bar\theta}_d\right|\right)^\intercal$ denotes the supremum and $S=\left\{\boldsymbol{\hat\theta}-\boldsymbol{\bar\theta}\in\mathbb{R}^d : \forall j, \left|\boldsymbol{\hat\theta}_j-\boldsymbol{\bar\theta}_j\right|\leq \boldsymbol{\xi}_j\right\}$ denotes the subspace bounded by the supremum, then the integral of the pdf $\int_S f(\boldsymbol{\hat\theta}-\boldsymbol{\bar\theta})d(\boldsymbol{\hat\theta}-\boldsymbol{\bar\theta})$ is the probability of the deviation within the supremum. Accordingly, the LDP mechanism with the highest probability is considered the best in high-dimensional space. Note that different supremum settings can lead to different winners. Next, we provide a case study to demonstrate how to benchmark Piecewise mechanism and Square wave mechanism by our framework.

\subsection{ A Case Study: How to Benchmark Piecewise Mechanism and Square Wave Mechanism in High-Dimensional Space?}
\label{subsec::case}
Since each dimension is perturbed equivalently in high-dimensional space, we study how to benchmark these two mechanisms in any single dimension. Suppose an original dataset with $d=100$ dimensions and $n=10000$ users, there are $v=10$ different original values $\left\{0.1, 0.2, 0.3, ..., 0.8, 0.9, 1.0\right\}$ in each dimension. For simplicity, we presume that the corresponding probability of each value in each dimension is $p=10\%$. For each user, she reports $m=100$ dimensions of her data to the data collector. As such, the data collector receives $r=\frac{nm}{d}=10000$ reports. Given the collective privacy budget $\epsilon=0.1$, each dimension is allocated $\epsilon/m=0.001$ privacy budget. Next, we demonstrate how to obtain the pdf in Theorem \ref{the::pdf} for each LDP mechanism. For Piecewise mechanism, we first obtain the variance of $\boldsymbol{t}^*_{ij}$:
\begin{equation}
\begin{aligned}
Var(\boldsymbol{t}^*_{ij})=&\mathbb{E}({\boldsymbol{t}^*_{ij}}^2)-\mathbb{E}^2(\boldsymbol{t}^*_{ij})\\=&\int_{-Q}^{l(\boldsymbol{t}^*_{ij})} \frac{(1-e^{-\epsilon/2m})x^2}{2e^{\epsilon/2m}+2} dx\\&+\int_{l(\boldsymbol{t}^*_{ij})}^{r(\boldsymbol{t}^*_{ij})} \frac{(e^{\epsilon/m}-e^{\epsilon/2m})x^2}{2e^{\epsilon/2m}+2}dx\\&+\int_{r(\boldsymbol{t}^*_{ij})}^{Q} \frac{(1-e^{-\epsilon/2m})x^2}{2e^{\epsilon/2m}+2} dx\\=&\frac{\boldsymbol{t}^*_{ij}}{e^{\epsilon/2m}-1}+\frac{e^{\epsilon/2m+3}}{3(e^{\epsilon/2m}-1)^2}
\label{equ::piecevar}
\end{aligned}
\end{equation}
We then derive the variance $\boldsymbol{\sigma}^2_j$ of Gaussian distribution that approximates $\boldsymbol{\hat\theta}_j-\boldsymbol{\bar\theta}_j$ according to Lemma \ref{lem::normal1}:
\begin{equation}
\begin{aligned}
\boldsymbol{\sigma}^2_j&={\frac{\sum_{z=1}^{v} pVar\left(\boldsymbol{t}^*_{(\sum_{o=1}^{z} r{p})j}\right)}{
 r}}\\&=\frac{\frac{10\%\times(0.1+0.2+...+1.0)}{e^{0.001/2}-1}+\frac{e^{0.001/2+3}}{3(e^{0.001/2}-1)^2}}{10000}\\&= 533.210
\end{aligned}
\label{equ::PieceNormal}
\end{equation}
Due to unbiased estimation, we can derive the pdf of $\boldsymbol{\hat\theta}_j-\boldsymbol{\bar\theta}_j$ in Piecewise mechanism by applying $d=1$, $\boldsymbol{\sigma}^2_j=533.210$, and $\boldsymbol{\delta}_{j}=0$ to Equation \ref{equ::pdf}:
\begin{equation}
\begin{aligned}
	f(\boldsymbol{\hat\theta}_j-\boldsymbol{\bar\theta}_j)=\frac{1}{57.900} \exp\left(-\frac{(\boldsymbol{\hat\theta}_j-\boldsymbol{\bar\theta}_j)^2}{1066.420}\right)
\end{aligned}
\end{equation}
For the Square wave mechanism, we have the bias of $\mathbb{E}(\boldsymbol{t}^*_{ij})$:
\begin{equation}
\begin{aligned}
\boldsymbol{\delta}_{ij}=&\mathbb{E}(\boldsymbol{t}^*_{ij}-\boldsymbol{t}_{ij})\\=&\int_{-b}^{\boldsymbol{t}_{ij}-b}\frac{x}{2be^{\epsilon/m}+1}dx+\int_{\boldsymbol{t}_{ij}-b}^{\boldsymbol{t}_{ij}+b} \frac{xe^{\epsilon/m}}{2be^{\epsilon/m}+1} dx\\&+\int_{\boldsymbol{t}_{ij}+b}^{1+b} \frac{x}{2be^{\epsilon/m}+1}dx-\boldsymbol{t}_{ij}\\=&\frac{2b(e^{\epsilon/m}-1)\boldsymbol{t}_{ij}}{2be^{\epsilon/m}+1}+\frac{1+2b}{2(2be^{\epsilon/m}+1)}-\boldsymbol{t}_{ij}
\end{aligned}
\end{equation}
and the variance of $\boldsymbol{t}^*_{ij}$:
\begin{equation}
\begin{aligned}
Var(\boldsymbol{t}^*_{ij})=&\mathbb{E}({\boldsymbol{t}^*_{ij}}^2)-\mathbb{E}^2(\boldsymbol{t}^*_{ij})\\=&\int_{-b}^{\boldsymbol{t}_{ij}-b}\frac{x^2}{2be^{\epsilon/m}+1}dx+\int_{\boldsymbol{t}_{ij}-b}^{\boldsymbol{t}_{ij}+b} \frac{x^2e^{\epsilon/m}}{2be^{\epsilon/m}+1} dx\\&+\int_{\boldsymbol{t}_{ij}+b}^{1+b} \frac{x^2}{2be^{\epsilon/m}+1}dx-(\boldsymbol{t}_{ij}+\boldsymbol{\delta}_{ij})^2\\&=\frac{b^2}{3}+\frac{(2b+1)(b+1-3\boldsymbol{t}_{ij}^2)}{3(2be^{\epsilon/m}+1)}-\boldsymbol{\delta}_{ij}^2-2\boldsymbol{\delta}_{ij}\boldsymbol{t}_{ij}
\end{aligned}
\end{equation}
We then derive the bias $\boldsymbol{\delta}_{j}$ and the variance $\boldsymbol{\sigma}_j^2$ of the Gaussian distribution that approximates $\boldsymbol{\hat\theta}_j-\boldsymbol{\bar\theta}_j$ according to Lemma \ref{lem::normal1}:
\begin{equation}
\begin{aligned}
\boldsymbol{\delta}_{j}&=\sum_{z=1}^{v}p \boldsymbol{\delta}_{ij}=-0.049\\
\boldsymbol{\sigma}^2_j&={\frac{\sum_{z=1}^{v} pVar\left(\boldsymbol{t}^*_{(\sum_{o=1}^{z} r{p})j}\right)}{
 r}}=3.365\times10^{-5}
 \end{aligned}
 \label{equ::SquareNormal}
\end{equation}
Finally, according to Theorem~\ref{the::pdf}, we can derive the pdf of $\boldsymbol{\hat\theta}_j-\boldsymbol{\bar\theta}_j$ in the Square wave mechanism by applying Equation \ref{equ::SquareNormal} and $d=1$ to Equation \ref{equ::pdf}:
\begin{equation}
\begin{aligned}
	f(\boldsymbol{\hat\theta}_j-\boldsymbol{\bar\theta}_j)=\frac{1}{0.015} \exp\left(- \frac{10^{5}(\boldsymbol{\hat\theta}_j-\boldsymbol{\bar\theta}_j+0.049)^2}{6.730}\right)
\end{aligned}
\end{equation}
Now that we have derived the pdf of $\boldsymbol{\hat\theta}_j-\boldsymbol{\bar\theta}_j$ in both LDP mechanisms, its integral $\int_{-\boldsymbol{\xi}_j}^{\boldsymbol{\xi}_j} f(\boldsymbol{\hat\theta}_j-\boldsymbol{\bar\theta}_j)d(\boldsymbol{\hat\theta}_j-\boldsymbol{\bar\theta}_j)$ is the probability that the deviation in $j$-th dimension is still within the supremum $\boldsymbol{\xi}_j=\sup|\boldsymbol{\hat\theta}_j-\boldsymbol{\bar\theta}_j|$. The higher probability the better the LDP mechanism. We vary $\boldsymbol{\xi}_j$ from 0.001 to 0.1 and show the resulted probabilities in Table \ref{tab::probability}. Piecewise mechanism is better than Square wave mechanism for smaller supremums (e.g., $0.001, 0.01$), which is mainly because Piecewise is an unbiased estimation while Square wave is not. However, if the supremum becomes larger (e.g., $0.05, 0.1$), in other words, if the collector can tolerate larger deviation, the Square wave mechanism is far better than the Piecewise mechanism because the variance of Gaussian distribution that approximates $\boldsymbol{\hat\theta}_j-\boldsymbol{\bar\theta}_j$ in the former is much smaller than that in the latter. That is to say, whether Piecewise or Square wave should be chosen depend on her tolerance of supremum $\boldsymbol{\xi}_j$.
\begin{table}[htbp]
\caption{Probabilities for the supremum to hold in one dimension}
\renewcommand{\arraystretch}{1.15} 
\begin{center}
\setlength{\tabcolsep}{1.1mm}{
	\begin{tabular}{|c|c|c|c|c|}\hline
		\textbf{$\boldsymbol{\xi}_j$}&$0.001$&$0.01$&$0.05$&$0.1$\\\hline
		\textbf{Piecewise}&$3.46\times 10^{-5}$&$3.46\times 10^{-4}$&$ 0.002$&$0.004$\\\hline
		\textbf{Square}&$2.12\times 10^{-16}$&$2.62\times 10^{-11}$&$ 0.644$&$1.000$\\\hline
		
	\end{tabular}}
\end{center}
\label{tab::probability}
\end{table}

\subsection{Approximation Error of Theorem \ref{the::pdf}}
\label{subsec::error}
Our analytical framework is based on one assumption that the data collector receives sufficiently large number of reports from users. Otherwise, the \emph{central limit theorem} provides an asymptotic approximation of the deviation. In order to find the gap between the approximated deviation and the true one, we study the approximation error of $\boldsymbol{\hat\theta}_j-\boldsymbol{\bar\theta}_j$ in terms of the number of reports $r_j$. 
Suppose the true $pdf$ of $\boldsymbol{\hat\theta}_j-\boldsymbol{\bar\theta}_j$ is $\bar{f}_j$, its corresponding cumulative distribution function ($cdf$) would be $\bar{F}_j(x)=\int_{-\infty}^x \bar{f}_j(\boldsymbol{\hat\theta}_j-\boldsymbol{\bar\theta}_j)d(\boldsymbol{\hat\theta}_j-\boldsymbol{\bar\theta}_j)$. According to Lemma \ref{lem::normal0} or Lemma \ref{lem::normal1}, the approximated $pdf$ of $\boldsymbol{\hat\theta}_j-\boldsymbol{\bar\theta}_j$ is $\hat{f}_j(\boldsymbol{\hat\theta}_j-\boldsymbol{\bar\theta}_j)=\frac{1}{\sqrt{2\pi}\boldsymbol{\sigma}_j}\exp(-\frac{\boldsymbol{(\hat\theta}_j-\boldsymbol{\bar\theta}_j-\boldsymbol{\delta}_j)^2}{2\boldsymbol{\sigma}_j^2})$, and its corresponding $cdf$ is $\hat{F}_j(x)=\int_{-\infty}^x \hat{f}_j(\boldsymbol{\hat\theta}_j-\boldsymbol{\bar\theta}_j)d(\boldsymbol{\hat\theta}_j-\boldsymbol{\bar\theta}_j)$. Then we have:
\begin{theorem}
For any LDP mechanism, the true $cdf$ $\bar{F}_j(x)$ and the approximated $cdf$ $\hat{F}_j(x)$ of $\boldsymbol{\hat\theta}_j-\boldsymbol{\bar\theta}_j$ differ by no more than $\frac{0.33554(\rho+0.415(r_j\boldsymbol{\sigma}_j)^3)}{r_j^{7/2}\boldsymbol{\sigma}_j^3}$, where $\rho=\mathbb{E}\left(\left|\boldsymbol{t}^*_{ij}-\boldsymbol{t}_{ij}-\boldsymbol{\delta}_{ij}\right|^3\right)$.
\end{theorem}
\begin{proof}
As necessary prerequisites, $\mathbb{E}(\boldsymbol{t}^*_{ij}-\boldsymbol{t}_{ij}-\boldsymbol{\delta}_{ij})=0$, and Lemma \ref{lem::normal0} and Lemma \ref{lem::normal1} prove that $\mathbb{E}((\boldsymbol{t}^*_{ij}-\boldsymbol{t}_{ij}-\boldsymbol{\delta}_{ij})^2)=\mathbb{E}(Var(\boldsymbol{t}^*_{ij}-\boldsymbol{t}_{ij}-\boldsymbol{\delta}_{ij})+\mathbb{E}^2(\boldsymbol{t}^*_{ij}-\boldsymbol{t}_{ij}-\boldsymbol{\delta}_{ij}))=\mathbb{E}(Var(\boldsymbol{t}^*_{ij}-\boldsymbol{t}_{ij}-\boldsymbol{\delta}_{ij}))=\mathbb{E}(Var(\boldsymbol{t}^*_{ij}))=(r_j\boldsymbol{\sigma}_j)^2$. Besides, we have to prove $\rho<\infty$. 
If $Bound(\mathcal{M})=1$, it surely establishes because $\boldsymbol{t}^*_{ij}$, $\boldsymbol{t}_{ij}$ and $\boldsymbol{\delta}_{ij}$ are all finite values in this case. If $Bound(\mathcal{M})=0$, we can prove that \emph{Laplace mechanism} satisfies this term. Note that $\boldsymbol{t}^*_{ij}-\boldsymbol{t}_{ij}-\boldsymbol{\delta}_{ij}=\boldsymbol{N}_{ij}$. Therefore, we have:
\begin{equation}
\begin{aligned}
\rho&=\mathbb{E}\left(\left|\boldsymbol{t}^*_{ij}-\boldsymbol{t}_{ij}-\boldsymbol{\delta}_{ij}\right|^3\right)\\&=\int_{-\infty}^{\infty} \left|x\right|^3Lap(\lambda=2m/\epsilon)dx\\&=\frac{1}{\lambda}\int_{0}^{\infty}x^3 \exp(-\frac{x}{\lambda})dx
\\&=\frac{3\lambda}{2}\mathbb{E}(x^2)=\frac{3\lambda}{2}2\lambda^2=3\lambda^3=\frac{24m^3}{\epsilon^3}\leq \infty
\end{aligned}
\label{equ::rho}
\end{equation}
As such, \emph{Berry–Esseen theorem} \cite{Korolev2010} establishes:
\begin{equation}
\begin{aligned}
\sup_{x \in \mathbb{R}}\left|\bar{F}_j(x)-\hat{F}_j(x)\right|
\leq \frac{0.33554(\rho+0.415(r_j\boldsymbol{\sigma}_j)^3)}{r_j^{7/2}\boldsymbol{\sigma}_j^3}
\end{aligned}
\end{equation}
\end{proof}
According to Lemma \ref{lem::normal0} and Lemma \ref{lem::normal1}, the value of $r_j\boldsymbol{\sigma}_j$ is irrelevant to $r_j$. Thus, $r_j\boldsymbol{\sigma}_j$ can be taken as a fixed value, which implies that the speed of convergence rate in our framework is at least on the order of $\frac{r_j^3}{r_j^{7/2}}=\frac{1}{\sqrt r_j}$. That is to say, the approximation error is still tolerable even if the number of reports is insufficient. We take \emph{Laplace mechanism} for example, where $\rho=3\lambda^3$ in Equation \ref{equ::rho} and $r_j\boldsymbol{\sigma}_j=\sqrt{Var(\boldsymbol{t}^*_{ij})}=\sqrt{(Var(Lap(\lambda)))}=\sqrt{2}\lambda$. Suppose the data collector only receives $r_j=1000$ reports, the approximation error between the true cdf and the approximated cdf  of $\boldsymbol{\hat\theta}_j-\boldsymbol{\bar\theta}_j$ is no more than $\frac{0.33554(\rho+0.415(r_j\boldsymbol{\sigma}_j)^3)}{r_j^{7/2}\boldsymbol{\sigma}_j^3}=\frac{0.33554\times(3\times\lambda^3+0.415\times2\times\sqrt{2}\times\lambda^3)}{2\times\sqrt{2}\times\lambda^3\times\sqrt{r_j}}\approx 1.57\%$.

	\section{HDR4ME: High-dimensional Re-calibration for Mean Estimation}
\label{sec::hdcsme}
In our analytical framework, we observe that dimensions $d$ has significant and direct influence on the deviation. In specific, $d$ dictates the privacy budget in each dimension, which directly affects the accuracy. In this section, we seize this opportunity to reduce the effective $d$ in the aggregation phase to improve the accuracy. The rationale of targeting at the aggregation phase instead of the perturbation or calibration is obvious --- the latter are mechanism-dependent whereas the former is universal to all LDP mechanisms. As such, our enhancement is orthogonal to all existing LDP optimizations.

In what follows, we first introduce {\it regularization} that can mitigate the negative influence in high dimensions. By integrating it into the aggregation, we propose a {\it re-calibration protocol} \emph{HDR4ME} and a solver algorithm based on proximal gradient descent. Last, we extend \emph{HDR4ME} for frequency estimation. Rigorous analysis is provided to prove its superiority over the existing one.

\subsection{Regularization: Diminishing Utility Deterioration in High-dimensional Space}
 {\it Regularization} is a common technique to re-calibrate the minimization tasks~\cite{boyd_vandenberghe_2004,buhlmann2011statistics,hoyer2004non,negahban2012unified,tibshirani1996regression}.
On the one hand, it directly reduces the dimensions $d$. On the other hand, it also reduces the scale of the perturbed data and thus diminishes the variance, which counteracts the utility deterioration caused by high dimensionality~\cite{duchi2018minimax}.

To explain regularization, let $\mathcal{L}(\boldsymbol\theta)$ denote a certain loss function regarding $\boldsymbol\theta \in \mathbb{R}^d$ while the regularization term is $\mathcal{R}(\boldsymbol\theta)$. $\mathcal{R}(\boldsymbol\theta)=\Vert\boldsymbol{\theta}\Vert_1$ and $\mathcal{R}=\Vert\boldsymbol{\theta}\Vert_2$ are the operators for $L_1$-regularization (abbreviated as $L_1$) and $L_2$-regularization (abbreviated as $L_2$), respectively. Figure \ref{fig::reg} illustrates the physical meaning of both regularizations in two dimensional space, where the black curves are isopleths of any loss function $\mathcal{L}(\boldsymbol\theta)$. The red square is the shape of $L_1$, while the blue circle is the shape of $L_2$. We notice that $\mathcal{L}(\boldsymbol\theta)$ converges to $\boldsymbol{\hat\theta}$ without regularization. In contrast to $\boldsymbol{\hat\theta}$, $\mathcal{L}(\boldsymbol\theta)$ tends to cross on coordinate axes with $L_1$ while it tends to cross on the circle with $L_2$. Let $\boldsymbol{\theta^*}$ denote the regularized results. Comparing both $\boldsymbol{\theta^*}$ with $\boldsymbol{\hat\theta}$, $L_1$ reduces both dimensions and the scale of $\boldsymbol{\hat\theta}$ while $L_2$ just reduces the scale of $\boldsymbol{\hat\theta}$. By integrating them in the aggregation phase as a re-calibration, we can mitigate the negative influence by high dimensionality. In the next subsection, we propose our re-calibration protocol $HDR4ME$.

\begin{figure}[htbp]
\begin{center}
\subfigure[$L_1$-regularization]{
\includegraphics[width=0.4\linewidth]{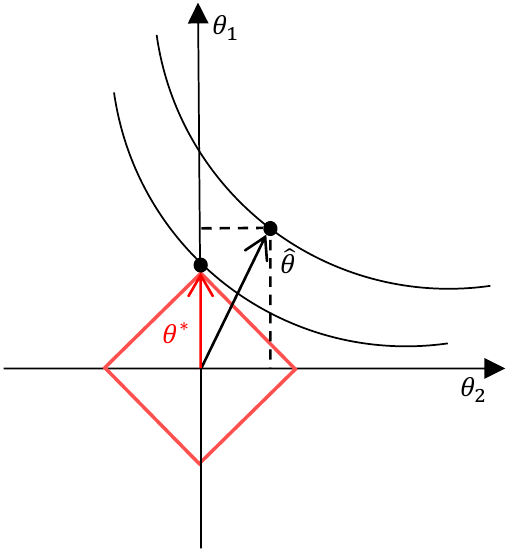}
}
\subfigure[$L_2$-regularization]{
\includegraphics[width=0.4\linewidth]{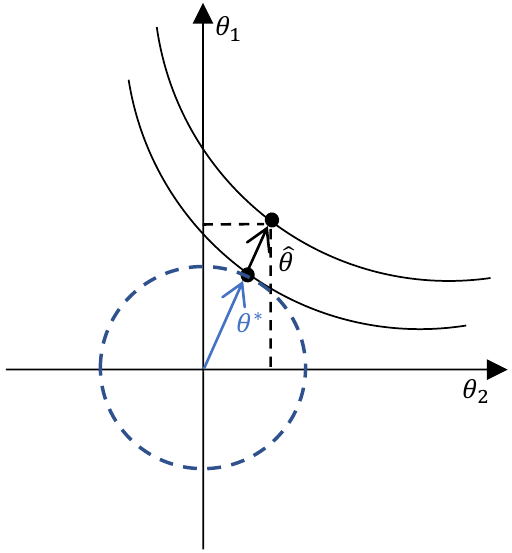}
}
\setlength{\abovecaptionskip}{0.01cm}
\caption{Regularization in two dimensions.}
\label{fig::reg}
\end{center}
\end{figure}

\subsection{HDR4ME---High Dimensional Re-calibration for Mean Estimation}
\label{hdr4me}

Recall that in each dimension, the data collector receives $r$ perturbed tuples $\left\{\boldsymbol{t}^{*}_i|1\leq i \leq r\right\}$, where $r=\frac{nm}{d}$. To add regularization terms, we first define the loss function of the aggregation $\mathcal{L}(\boldsymbol{\theta})=\frac{1}{2r}\sum_{i=1}^r \left\|\boldsymbol{t}^{*}_i-\boldsymbol{\theta}\right\|^2_2$. On this basis, we add regularization terms $\mathcal{R}(\boldsymbol{\theta})$ to $\mathcal{L}(\boldsymbol{\theta})$ to obtain the enhanced mean $\boldsymbol{\theta^*}$ as follows:
\begin{equation}
\boldsymbol{\theta}^*=\arg\min_{\theta\in \mathbb{R}^d}\left\{\mathcal{L}(\boldsymbol{\theta})+\mathcal{R}(\boldsymbol{\lambda}^*\circ\boldsymbol\theta)\right\},
\end{equation}
where $\mathcal{R}(\boldsymbol{\theta})=\Vert \boldsymbol{\theta} \Vert_1$ or $\Vert \boldsymbol{\theta} \Vert_2$ and $\boldsymbol{\lambda}^*=(\boldsymbol{\lambda}_1^*, ..., \boldsymbol{\lambda}^*_d)^\intercal$ is the regularization weight (which controls the degree of the involvement of regularization). In particular, $\boldsymbol{\lambda}^*\circ\boldsymbol\theta=(\boldsymbol{\lambda}_1^*\boldsymbol\theta_1, ..., \boldsymbol{\lambda}^*_d\boldsymbol\theta_d)^\intercal$ is {\it Hadamard product}. In what follows, we provide detailed utility analysis of $HDR4ME$ with $L_1$- and $L_2$-regularization, respectively, together with the specification of $\boldsymbol{\lambda}^*$.

\noindent
\textbf{HDR4ME with $L_1$-regularization.} With this re-calibration, the deviation $\Vert\boldsymbol{\hat\theta}-\boldsymbol{\bar\theta}\Vert_2$ can be significantly reduced by dimensionality and perturbation reduction. The following lemma discusses the suitable choice of $\boldsymbol{\lambda}^*$ and the threshold for utility enhancement.


\begin{lemma}
	\emph{HDR4ME} with $L_1$-regularization can improve accuracy in $j$-th dimension if
	\begin{equation}
		\begin{aligned}			
		\boldsymbol{\lambda}_j^*=\sup\left|\boldsymbol{\hat\theta}_j-\boldsymbol{\bar\theta}_j\right| \ \mbox{and} \ \left|\boldsymbol{\hat\theta}_j-\boldsymbol{\bar\theta}_j\right|> 1
		\end{aligned}
	\end{equation}
where $\boldsymbol{\hat\theta}_j-\boldsymbol{\bar\theta}_j$ is obtained from Lemma \ref{lem::normal0} or Lemma \ref{lem::normal1}.
\label{lem::L1para}
\end{lemma}
\begin{proof}
Since $\Vert\boldsymbol\theta\Vert_1$ is non-differentiable, we adopt an alternative solution, namely, proximal gradient descent (PGD) \cite{boyd_vandenberghe_2004,nitanda2014stochastic,li2015accelerated}. Our objective is to obtain the iterative equation to solve our protocol. First, we get the derivative of  $\mathcal{L}(\boldsymbol{\theta})$:
\begin{equation}
\begin{aligned}
\nabla\mathcal{L}(\boldsymbol{\theta})=\frac{1}{r}\sum_{i=1}^{r}\left(\boldsymbol{\theta}-\boldsymbol{t_i^*}\right)=\boldsymbol{\theta}-\frac{1}{r}\sum_{i=1}^{r} \boldsymbol{t}_i^*=\boldsymbol{\theta}-\boldsymbol{\hat\theta}
\end{aligned}
\end{equation}
Thus, the derivative of $\nabla\mathcal{L}(\boldsymbol{\theta})$ is
$\frac{\mathrm{d}\nabla \mathcal{L}(\boldsymbol{\theta})}{\mathrm{d}{\boldsymbol{\theta}}}= 1$.
According to \emph{Cauchy mean value theorem}, we have:
\begin{equation}
\Vert \nabla \mathcal{L}(\boldsymbol{\theta})-\nabla \mathcal{L}(\boldsymbol{\theta}_k) \Vert_2^2 \leq \Vert \boldsymbol{\theta}-\boldsymbol{\theta}_k\Vert_2^2,
\label{equ::lip}
\end{equation}
where $\boldsymbol{\theta}_k$ is the result of $k$-th iteration. By
\emph{second-order Taylor expansion} around $\boldsymbol{\theta}_k$, we get:
\begin{equation}
\begin{aligned}
\mathcal{L}(\boldsymbol{\theta})&\cong \mathcal{L}(\boldsymbol{\theta}_k)+\left \langle \nabla \mathcal{L}(\boldsymbol{\theta}_k),\boldsymbol{\theta}- \boldsymbol{\theta}_k\right \rangle+\frac{1}{2}\Vert \boldsymbol{\theta}-\boldsymbol{\theta}_k \Vert^2
\\&=\frac{1}{2} \left \| \boldsymbol{\theta}-\left(\boldsymbol{\theta}_k-\nabla \mathcal{L}(\boldsymbol{\theta}_k)\right) \right \|^2_2+constant
\label{equ::taylor}
\end{aligned}
\end{equation}
To minimize the loss function $\mathcal{L}$, we get the iterative equation $\boldsymbol{\theta}_{k+1}=\boldsymbol{\theta}_k-\nabla \mathcal{L}(\boldsymbol{\theta}_k)$.
We then introduce $L_1$-regularization term into the iteration:
\begin{equation}
\boldsymbol{\theta}_{k+1}=\arg\min_{\boldsymbol{\theta}} \frac{1}{2} \left \| \boldsymbol{\theta}-\left(\boldsymbol{\theta}_k-\nabla \mathcal{L}(\boldsymbol{\theta}_k)\right) \right \|^2_2+\Vert\boldsymbol{\lambda}^*\circ\boldsymbol\theta\Vert_1
\end{equation}
Since each dimension is independent of each other, we have the following solution for each dimension:
\begin{equation}
(\boldsymbol{\theta}_{k+1})_j=\arg\min_{\boldsymbol{\theta}_j} \frac{1}{2} \left| \boldsymbol{\theta}_j-\left((\boldsymbol{\theta}_k)_j-\nabla \mathcal{L}(\boldsymbol{\theta}_k)_j\right) \right|^2_2+\left|\boldsymbol{\lambda}_j^*\boldsymbol\theta_j\right|_1
\label{equ::prox}
\end{equation}
As such, how to compute $(\boldsymbol{\theta}_{k+1})_j$ really depends on whether $\boldsymbol{\theta}_j$ is positive, zero or negative ($\boldsymbol{\lambda}_j$ is positive). In particular, we let $\boldsymbol{z}=\boldsymbol{\theta}_k-\nabla \mathcal{L}(\boldsymbol{\theta}_k)=\boldsymbol{\theta}_k-\boldsymbol{\theta}_k+\boldsymbol{\hat\theta}=\boldsymbol{\hat\theta}$. If $\boldsymbol{\theta}_j>0$, we get the gradient of Equation \ref{equ::prox} as $\boldsymbol{\theta}_j-\boldsymbol{z}_j+\boldsymbol{\lambda}_j^*$. By making it zero, we obtain $(\boldsymbol{\theta}_{k+1})_j=\boldsymbol{z}_j-\boldsymbol{\lambda}_j^*>0$, in which case $\boldsymbol{z}_j>\boldsymbol{\lambda}_j^*$. If $\boldsymbol{\theta}_j<0$, we similarly obtain $(\boldsymbol{\theta}_{k+1})_j=\boldsymbol{z}_j+\boldsymbol{\lambda}_j^*<0$, in which case $\boldsymbol{z}_j<-\boldsymbol{\lambda}_j^*$. If $\boldsymbol{\theta}_j=0$, Equation \ref{equ::prox} simply converges and $(\boldsymbol{\theta}_{k+1})_j=0$, which corresponds with $\left|\boldsymbol{z}_j \right| \leq \boldsymbol{\lambda}^*_j$. Accordingly, we have the following iteration:

\begin{equation}
(\boldsymbol{\theta}_{k+1})_j=
\begin{cases}
\boldsymbol{z}_j-\boldsymbol{\lambda}^*_j, \qquad \boldsymbol{z}_j>\boldsymbol{\lambda}^*_j \\
0, \qquad \qquad \ \ \left|\boldsymbol{z}_j \right| \leq \boldsymbol{\lambda}^*_j\\
\boldsymbol{z}_j+\boldsymbol{\lambda}^*_j, \qquad \boldsymbol{z}_j<-\boldsymbol{\lambda}^*_j
\end{cases}
\label{equ::iteration}
\end{equation}
Since $\boldsymbol{z}$ and $\boldsymbol{\lambda}^*$ are deterministic, Equation \ref{equ::iteration} is actually a one-off solver. If we set $\boldsymbol{\lambda}_j^*=\sup\left|\boldsymbol{\hat\theta}_j-\boldsymbol{\bar\theta}_j\right|$, for $\boldsymbol{\theta}_j^*>0$, we have $\boldsymbol{\theta}_j^*=\boldsymbol{z}_j-\boldsymbol{\lambda}_j^*=\boldsymbol{\hat\theta}_j-\sup\left|\boldsymbol{\hat\theta}_j-\boldsymbol{\bar\theta}_j\right|\leq\boldsymbol{\hat\theta}_j-\left|\boldsymbol{\hat\theta}_j-\boldsymbol{\bar\theta}_j\right|$. Since $\boldsymbol{\theta}_j^*$ is re-calibrated from $\boldsymbol{z}_j=\boldsymbol{\hat\theta}_j$, $\boldsymbol{\hat\theta}_j$ has the same sign as $\boldsymbol{\theta}_j^*$, which implies $\boldsymbol{\hat\theta}_j>\left|\boldsymbol{\hat\theta}_j-\boldsymbol{\bar\theta}_j\right|$. Suppose $\left|\boldsymbol{\hat\theta}_j-\boldsymbol{\bar\theta}_j\right|>1$, which happens frequently in high-dimensional space,  we then have $\boldsymbol{\hat\theta}_j>1$. Because $\boldsymbol{\bar\theta}_j\leq 1$, $\left|\boldsymbol{\hat\theta}_j-\boldsymbol{\bar\theta}_j\right|=\boldsymbol{\hat\theta}_j-\boldsymbol{\bar\theta}_j$ holds. Therefore, we have $0<\boldsymbol{\theta}_j^*\leq \boldsymbol{\hat\theta}_j-\left|\boldsymbol{\hat\theta}_j-\boldsymbol{\bar\theta}_j\right|=\boldsymbol{\bar\theta}_j\leq 1$, which proves $0\leq\left|\boldsymbol{\theta}^*_j-\boldsymbol{\bar\theta}_j\right|< 1$. As such,
$\left|\boldsymbol{\theta}^*_j-\boldsymbol{\bar\theta}_j\right|<1<\left|\boldsymbol{\hat\theta}_j-\boldsymbol{\bar\theta}_j\right|$ accordingly holds. Similarly, we derive $\left|\boldsymbol{\theta}^*_j-\boldsymbol{\bar\theta}_j\right|<1<\left|\boldsymbol{\hat\theta}_j-\boldsymbol{\bar\theta}_j\right|$ for $\boldsymbol{\theta}_j^*<0$. For $\boldsymbol{\theta}_j^*=0$, $\left|\boldsymbol{\theta}^*_j-\boldsymbol{\bar\theta}_j\right|=\left|\boldsymbol{\bar\theta}_j\right|\leq1<\left|\boldsymbol{\hat\theta}_j-\boldsymbol{\bar\theta}_j\right|$. In general, we have:
\begin{equation}
\begin{aligned}
\left|\boldsymbol{\theta}^*_j-\boldsymbol{\bar\theta}_j\right|&<\left|\boldsymbol{\hat\theta}_j-\boldsymbol{\bar\theta}_j\right|\\\quad\mbox{if}\ \boldsymbol{\lambda}_j^*=\sup\left|\boldsymbol{\hat\theta}_j-\boldsymbol{\bar\theta}_j\right| \ &\mbox{and}\ \left|\boldsymbol{\hat\theta}_j-\boldsymbol{\bar\theta}_j\right|> 1
\label{equ::L1hold}
\end{aligned}
\end{equation}
\end{proof}
This lemma specifies the suitable regularization weight and the required threshold for utility enhancement in one dimension. On this basis, we prove the superiority of HDR4ME with $L_1$ to the current aggregation in high-dimensional space.
\begin{theorem}
For any high-dimensional LDP mechanism $\mathcal{M}$ under \emph{HDR4ME} with $L_1$-regularization, the following inequality holds with at least $1-\int_{-1}^{1}...\int_{-1}^{1} f(\boldsymbol{\hat\theta}-\boldsymbol{\bar\theta})d(\boldsymbol{\hat\theta}-\boldsymbol{\bar\theta})$ probability:
\begin{equation}
\Vert\boldsymbol{\theta^*}-\boldsymbol{\bar\theta}\Vert_2 <\Vert\boldsymbol{\hat\theta}-\boldsymbol{\bar\theta}\Vert_2
\end{equation}
\label{the::L1}
where $f(\boldsymbol{\hat\theta}-\boldsymbol{\bar\theta})$ is obtained from Theorem~\ref{the::pdf}.
\end{theorem}
\begin{proof}
Lemma \ref{lem::L1para} establishes that $\left|\boldsymbol{\theta}^*_j-\boldsymbol{\bar\theta}_j\right|<\left|\boldsymbol{\hat\theta}_j-\boldsymbol{\bar\theta}_j\right|$ holds with one certain threshold: $\left|\boldsymbol{\hat\theta}_j-\boldsymbol{\bar\theta}_j\right|> 1$. Theorem \ref{the::pdf} derives that $\left|\boldsymbol{\hat\theta}_j-\boldsymbol{\bar\theta}_j\right|> 1$ holds for $\forall j\in [1,d]$ with at least the probability $1-\int_{-1}^{1}...\int_{-1}^{1} f(\boldsymbol{\hat\theta}-\boldsymbol{\bar\theta})d(\boldsymbol{\hat\theta}-\boldsymbol{\bar\theta})$. On this basis,
\begin{equation}
\Vert\boldsymbol{\theta^*}-\boldsymbol{\bar\theta}\Vert_2=\sqrt{\sum_{j=1}^d \left|\boldsymbol{\theta}^*_j-\boldsymbol{\bar\theta}_j\right|^2}< \sqrt{\sum_{j=1}^d \left|\boldsymbol{\hat\theta}_j-\boldsymbol{\bar\theta}_j\right|^2}=\Vert\boldsymbol{\hat\theta}-\boldsymbol{\bar\theta}\Vert_2
\end{equation}
\end{proof}

In general, Theorem \ref{the::L1} derives the least probability for $L_1$ to enhance utilities in high-dimensional space. Nevertheless, a solver to HDR4ME with $L_1$ is still required. Applying $\boldsymbol{z}=\boldsymbol{\hat\theta}$ and $(\boldsymbol{\theta}_{k+1})_j=\boldsymbol{\theta}^*_j$ to Equation \ref{equ::iteration}, we have:

\begin{equation}
\boldsymbol{\theta}^*_j=
\begin{cases}
\boldsymbol{\hat\theta}_j-\boldsymbol{\lambda}^*_j, \qquad \boldsymbol{\hat\theta}_j>\boldsymbol{\lambda}^*_j \\
0, \qquad \qquad \ \ \left|\boldsymbol{\hat\theta}_j \right| \leq \boldsymbol{\lambda}^*_j\\
\boldsymbol{\hat\theta}_j+\boldsymbol{\lambda}^*_j, \qquad \boldsymbol{\hat\theta}_j<-\boldsymbol{\lambda}^*_j
\end{cases}
\label{equ::L1solution}
\end{equation}

Equation \ref{equ::L1solution} is a one-off, non-iterative solver for HDR4ME with $L_1$, which simply re-calibrates the estimated mean to get the enhanced mean. As such, the data collector can enhance utilities without bearing extra computational burden. 

\noindent
\textbf{HDR4ME with $L_2$-regularization.} This re-calibration can obtain much deviation $\Vert\boldsymbol{\hat\theta}-\boldsymbol{\bar\theta}\Vert_2$ by scale reduction. To achieve this, $\boldsymbol{\lambda}^*$ must satisfy the following condition.
\begin{lemma}
\emph{HDR4ME} with $L_2$-regularization can improve accuracy in $j$-th dimension if 
	\begin{equation}
		\begin{aligned}			\boldsymbol{\lambda}^*_j=\sup\frac{\boldsymbol{\hat\theta}_j-\boldsymbol{\bar\theta}_j}{2\boldsymbol{\bar\theta}_j} \ \mbox{and}\ \left|\boldsymbol{\hat\theta}_j-\boldsymbol{\bar\theta}_j\right|> 2
		\end{aligned}
	\end{equation}
where $\boldsymbol{\hat\theta}_j-\boldsymbol{\bar\theta}_j$ is obtained from Lemma \ref{lem::normal0} or Lemma \ref{lem::normal1}, and $\boldsymbol{\bar\theta}_j$ can select the mean of the normal distribution that approximates $\boldsymbol{\hat\theta}_j-\boldsymbol{\bar\theta}_j$ in our framework.
	\label{lem::L2para}
\end{lemma}
 \begin{proof}
Following Equation \ref{equ::taylor}, we add $L_2$-regularization term into the iteration:
\begin{equation}
\boldsymbol{\theta}_{k+1}=\arg\min_{\boldsymbol{\theta}} \frac{1}{2} \left \| \boldsymbol{\theta}-\left(\boldsymbol{\theta}_k-\nabla \mathcal{L}(\boldsymbol{\theta}_k)\right) \right \|^2_2+\Vert\boldsymbol{\lambda}^*\circ\boldsymbol\theta\Vert_2^2
\end{equation}
Since each dimension is perturbed independently, we have the following solution for each dimension:
\begin{equation}
(\boldsymbol{\theta}_{k+1})_j=\arg\min_{\boldsymbol{\theta}_j} \frac{1}{2} \left| \boldsymbol{\theta}_j-\left((\boldsymbol{\theta}_k)_j-\nabla \mathcal{L}(\boldsymbol{\theta}_k)_j\right) \right|^2_2+\left|\boldsymbol{\lambda}_j^*\boldsymbol\theta_j\right|^2
\label{equ::grad}
\end{equation}
where $\boldsymbol{z}=\boldsymbol{\theta}_k-\nabla \mathcal{L}(\boldsymbol{\theta}_k)=\boldsymbol{\hat\theta}$.
Note that Equation \ref{equ::grad} is differentiable. Applying 0 to the derivative of Equation \ref{equ::grad}, we have $\boldsymbol{\theta}_j^*=\frac{\boldsymbol{\hat\theta}_j}{2\boldsymbol{\lambda}^*_j+1}$. If $\boldsymbol{\lambda}^*_j=\sup\frac{\boldsymbol{\hat\theta}_j-\boldsymbol{\bar\theta}_j}{2\boldsymbol{\bar\theta}_j}$, our framework implies $\boldsymbol{\lambda}^*_j>0$. Then, we derive:
\begin{equation}
\begin{aligned}
|\boldsymbol{\theta}_j^*|=\left|\frac{\boldsymbol{\hat\theta}_j}{2\boldsymbol{\lambda}^*_j+1}\right|&=\left|\frac{\boldsymbol{\hat\theta}_j}{\sup\frac{\boldsymbol{\hat\theta}_j-\boldsymbol{\bar\theta}_j}{\boldsymbol{\bar\theta}_j}+1}\right|\leq\left|\frac{\boldsymbol{\hat\theta}_j}{\frac{\boldsymbol{\hat\theta}_j-\boldsymbol{\bar\theta}_j}{\boldsymbol{\bar\theta}_j}+1}\right|=\left|\boldsymbol{\bar\theta}_j\right|
\end{aligned}
\end{equation}
Therefore, we have $0\leq|\boldsymbol{\theta}_j^*|\leq\left|\boldsymbol{\bar\theta}_j\right|\leq1$, which implies $0\leq\left|\boldsymbol{\theta}^*_j-\boldsymbol{\bar\theta}_j\right|\leq 2$. Namely, we have:
\begin{equation}
\begin{aligned}
\left|\boldsymbol{\theta}^*_j-\boldsymbol{\bar\theta}_j\right|&<\left|\boldsymbol{\hat\theta}_j-\boldsymbol{\bar\theta}_j\right|\\\mbox{if}\ 	\boldsymbol{\lambda}^*_j=\sup\frac{\boldsymbol{\hat\theta}_j-\boldsymbol{\bar\theta}_j}{2\boldsymbol{\bar\theta}_j} \ &\mbox{and}\ \left|\boldsymbol{\hat\theta}_j-\boldsymbol{\bar\theta}_j\right|> 2
\label{equ::L2hold}
\end{aligned}
\end{equation}
\end{proof}

Now that this lemma specifies the suitable regularization weight and the required threshold for utility enhancement in one dimension, we further prove the superiority of HDR4ME with $L_2$ to the current aggregation in high-dimensional space.
\begin{theorem}
For any high-dimensional LDP mechanism $\mathcal{M}$ under \emph{HDR4ME} with $L_2$-regularization, the following inequality holds with at least $1-\int_{-2}^{2}...\int_{-2}^{2} f(\boldsymbol{\hat\theta}-\boldsymbol{\bar\theta})d(\boldsymbol{\hat\theta}-\boldsymbol{\bar\theta})$ probability:\begin{equation}
\Vert\boldsymbol{\theta^*}-\boldsymbol{\bar\theta}\Vert_2 <\Vert\boldsymbol{\hat\theta}-\boldsymbol{\bar\theta}\Vert_2
\end{equation}
\label{the::L2}
where $f(\boldsymbol{\hat\theta}-\boldsymbol{\bar\theta})$ is obtained from Theorem~\ref{the::pdf}.
\end{theorem}
\begin{proof}
Equation \ref{equ::L1hold} in Lemma \ref{lem::L2para} derives that $\left|\boldsymbol{\theta}^*_j-\boldsymbol{\bar\theta}_j\right|<\left|\boldsymbol{\hat\theta}_j-\boldsymbol{\bar\theta}_j\right|$ holds with one certain threshold: $\left|\boldsymbol{\hat\theta}_j-\boldsymbol{\bar\theta}_j\right|> 2$. Theorem \ref{the::pdf} derives that $\left|\boldsymbol{\hat\theta}_j-\boldsymbol{\bar\theta}_j\right|>2$ holds for $\forall j\in [1,d]$ with at least the probability $1-\int_{-2}^{2}...\int_{-2}^{2} f(\boldsymbol{\hat\theta}-\boldsymbol{\bar\theta})d(\boldsymbol{\hat\theta}-\boldsymbol{\bar\theta})$. On this basis,
\begin{equation}
\Vert\boldsymbol{\theta^*}-\boldsymbol{\bar\theta}\Vert_2=\sqrt{\sum_{j=1}^d \left|\boldsymbol{\theta}^*_j-\boldsymbol{\bar\theta}_j\right|^2}< \sqrt{\sum_{j=1}^d \left|\boldsymbol{\hat\theta}_j-\boldsymbol{\bar\theta}_j\right|^2}=\Vert\boldsymbol{\hat\theta}-\boldsymbol{\bar\theta}\Vert_2
\end{equation}
\end{proof}

With our framework, Theorem \ref{the::L2} derives the least probability for $L_2$ to enhance utilities in high-dimensional space, in which case the enhanced mean is always better than estimated mean. To solve HDR4ME with $L_2$, we compute the derivative of Equation \ref{equ::grad} and set it to zero:
\begin{equation}
\boldsymbol{\theta}_j^*=\boldsymbol{\theta}_k-\nabla \mathcal{L}(\boldsymbol{\theta}_k)=\frac{\boldsymbol{\hat\theta}_j}{2\boldsymbol{\lambda}^*_j+1}
\end{equation}
Similarly, the above is also a one-off, non-iterative solver for HDR4ME with $L_2$, which does not increase the computational burden of the data collector.

As a final note, both types of {\it HDR4ME} are designed for ``high-dimensional'' space only. In such a space, the useful statistics are flooded by much larger noise, which provides us room to make utility enhancement. If the number of dimensions is not high or the collective privacy budget is rather large, which generally means that the threshold for either regularization to enhance utilities is not reached, our re-calibration can be harmful. 
\subsection{High-dimensional Re-calibration for Frequency Estimation}
For various LDP mechanisms, high-dimensional frequency estimation is never sufficiently discussed, especially when some mechanisms claim to be applicable to both mean and frequency estimations \cite{wang2017locally,wang2019collecting}. As such, we also generalize our re-calibration to frequency estimation. Note that any categorical value can be mapped into a binary vector with {\it histogram encoding} \cite{wang2017locally}. Suppose there are $d$ categorical dimensions and $v_j (1\leq j\leq d)$ categories in each dimension, any categorical value in $j$-th dimension $\boldsymbol{t}_{ij} (1\leq i\leq r_j)$ can be encoded to a $v_j$-entry vector $(0.0, 0.0, ..., 1.0, ..., 0.0)^\intercal$ with only the $\boldsymbol{t}_{ij}$-th entry to be 1.0. As such, each of $d$ categorical dimensions is expanded to one $v_j$-dimensional numerical space. Note that each entry of encoded vectors ranges from $[0,1]$. If each user reports $m$ dimensions of her perturbed data to the data collector, the collective $\epsilon$-LDP can be guaranteed by applying $\frac{\epsilon}{2m}$ to each entry of vectors \cite{wang2017locally} regardless of LDP mechanisms. As such, the data collector receives $r_j$ $v_j$-entry perturbed vectors in $j$-th dimension. Since each entry corresponds with one certain categorical value, the mean of $r_j$ perturbed vectors corresponds with the estimated frequencies in $j$-th dimension, with each entry of the mean to be the frequency of each categorical value. In general, we can convert one $d$-dimensional frequency estimation to $d$ high-dimensional mean estimation tasks. On this basis, both our framework and re-calibration protocol can further apply. 

	\section{Experimental Evaluation}
\label{sec::experiment}
To verify both the analytical framework and the re-calibration protocol, we conduct experiments under a real dataset COV-19\footnote{https://www.kaggle.com/allen-institute-for-ai/CORD-19-research-challenge} and three synthetically distributed datasets, namely Gaussian, Poisson and Uniform. The following are some descriptions of four datasets:
\begin{enumerate}
\item The COV-19 dataset consists of 150,000 users and 750 dimensions, where each dimension has high correlations with others. 

\item The Gaussian dataset consists of tunable users and dimensions. The standard deviation of all dimensions is set to $1/16$. 10\% dimensions have their mathematical expectations $\mu=0.9$ whereas the other 90\% have $\mu=0$.

\item The Poisson dataset consists of 150,000 users and 300 dimensions, where each dimension follows a Poisson distribution with a random expectation from $1$ to $99$.

\item The Uniform dataset consists of tunable users and dimensions. 

\end{enumerate}

The aims of our experiments are twofold. First, we confirm the effectiveness of our analytical framework, namely, $\boldsymbol{\hat\theta}_j-\boldsymbol{\bar\theta}_j$ can be approximated with one certain Gaussian distribution. Second, we compare the performances of \emph{HDR4ME} on top of the aggregation results of three state-of-the-art LDP mechanisms, i.e., Laplace \cite{10.1007/11681878_14}, Piecewise \cite{wang2019collecting}, and Square wave \cite{li2020estimating}. Each dimension is normalized into $[-1,1]$, and each experiment is repeated 100 times to obtain the averaged result unless otherwise indicated. All our experiments are implemented in MATLAB on a laptop computer with Intel Core i7-10750H 2.59 GHz CPU, 32G RAM on Windows 10 operation system.

In the first set of experiments, for a start, we use Uniform dataset to verify the effectiveness of our analytical framework. In specific, we set 200,000 users and 5,000 dimensions. For each user, they send 50 dimensions of her perturbed tuples to the data collector. To ensure generality, we conduct experiments on Laplace, Piecewise and Square wave, respectively. Each experiment is iterated 1,000 times, and we collect the means of 1,000 times in the first dimension. Given the collective privacy budget $\epsilon=1$, Fig. \ref{fig::CLT} shows how our framework models the means from experiments. In each sub-figure, the blue line is the pdf of the deviations from our framework while the orange squares are the pdf estimate from experiments. In all three mechanisms, our framework effectively approximates experimental results. Recall that we provide a case study in Section \ref{subsec::case} to benchmark Piecewise and Square wave. To support the benchmark results, we discretize the Uniform dataset and conduct experiments in Fig. \ref{fig::bench}. In both mechanisms, the pdf functions computed in our case study perfectly model the experimental results, which confirms the effectiveness of the benchmark by our framework.
\begin{figure}[htbp]
\centering
\subfigure[Laplace]{
\includegraphics[width=0.465\linewidth]{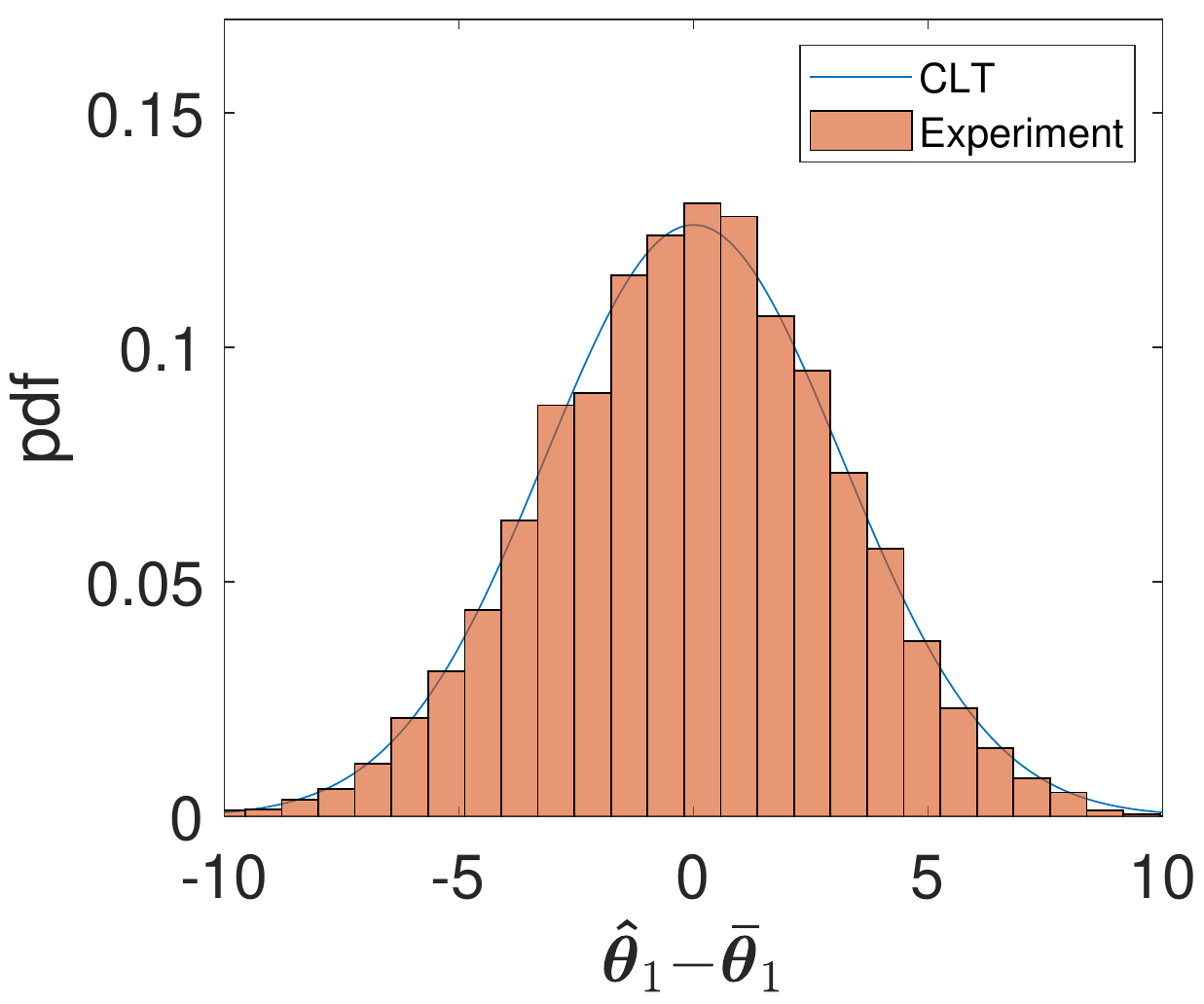}
}
\subfigure[Piecewise]{
\includegraphics[width=0.465\linewidth]{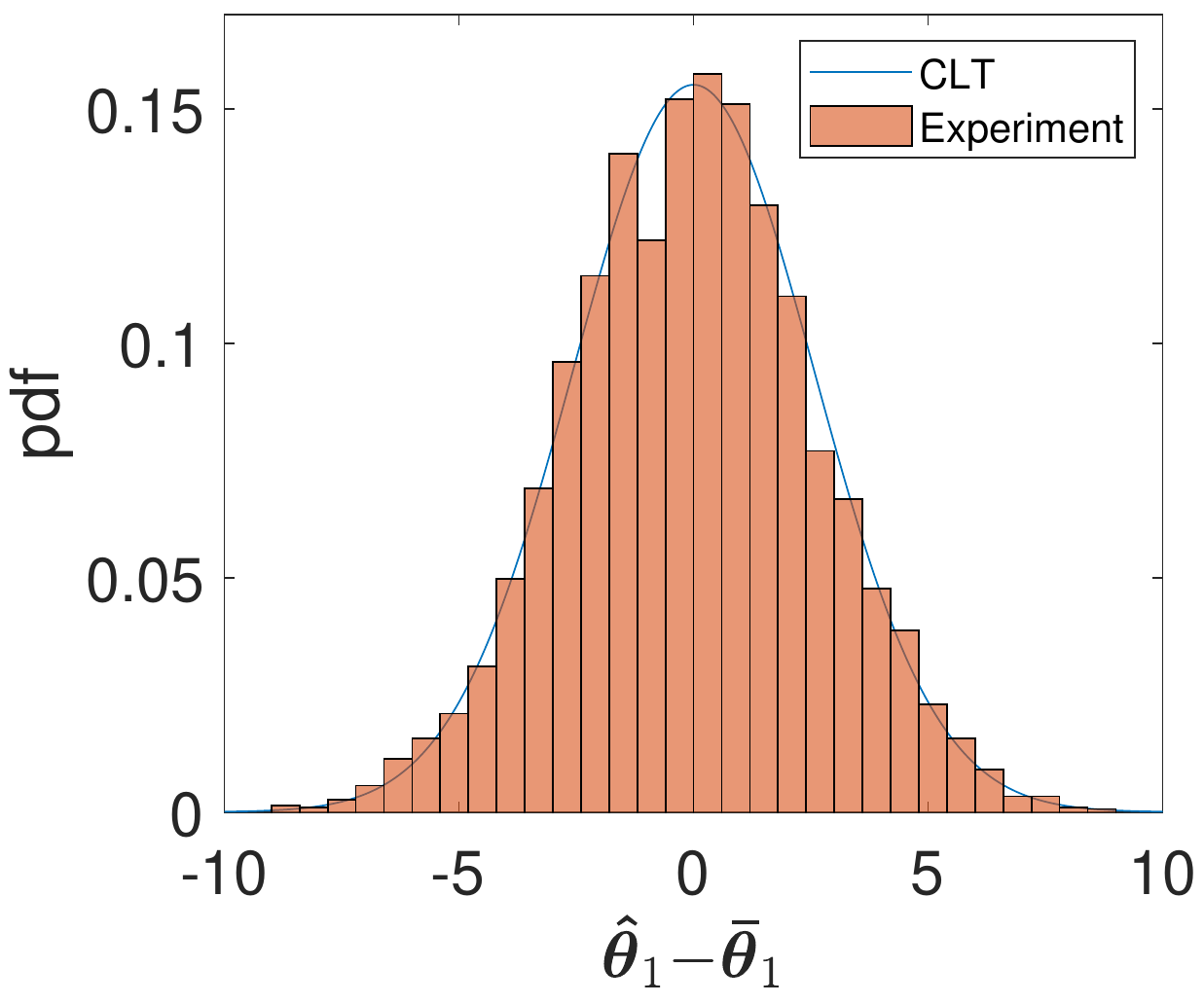}
}
\\
\subfigure[Square]{
\includegraphics[width=0.46\linewidth]{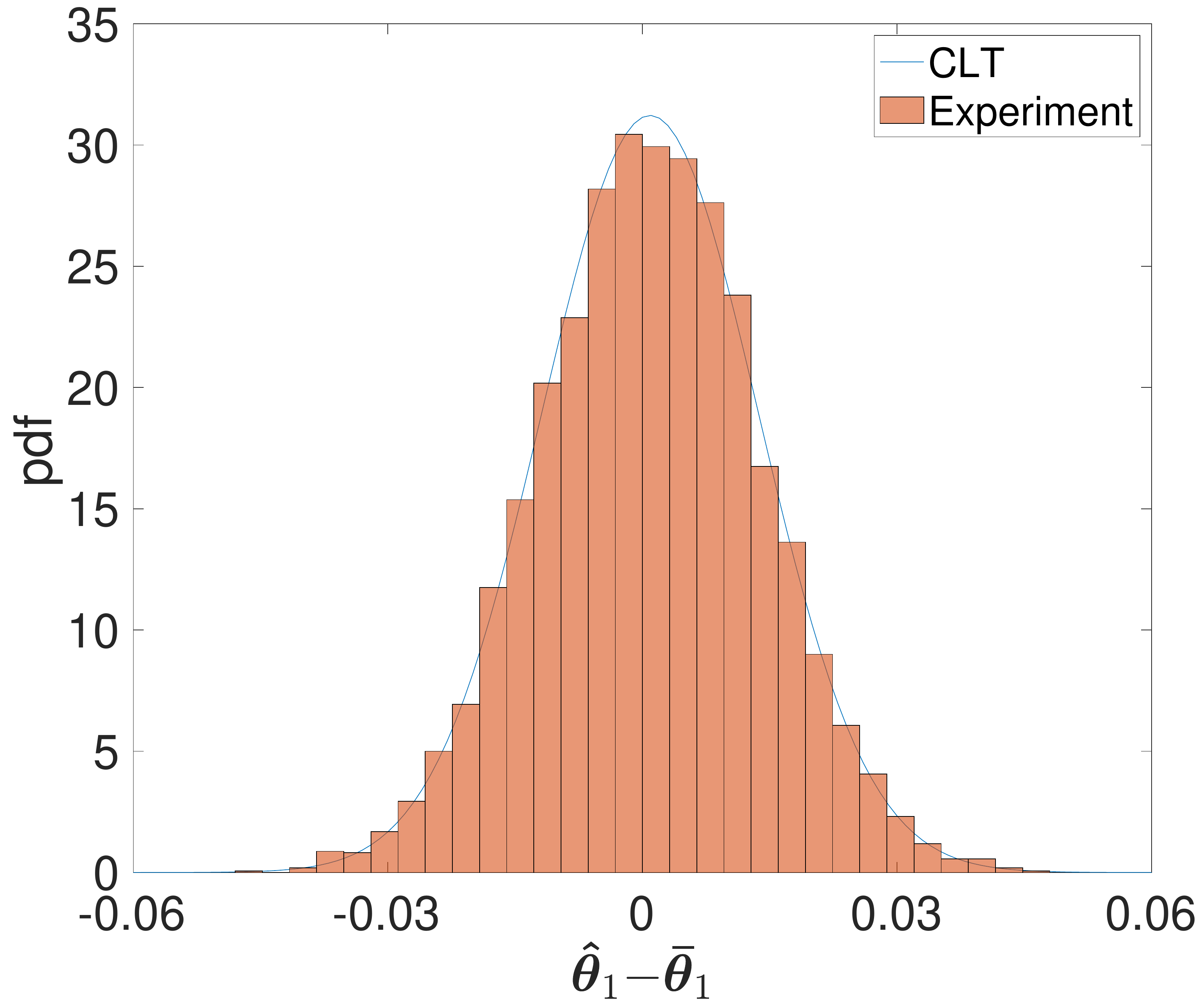}
}
\caption{Analysis vs. experimental results on Uniform dataset (d=5,000).}
\label{fig::CLT}
\end{figure}

\begin{figure}[htbp]
\centering
\subfigure[Piecewise]{
\includegraphics[width=0.465\linewidth]{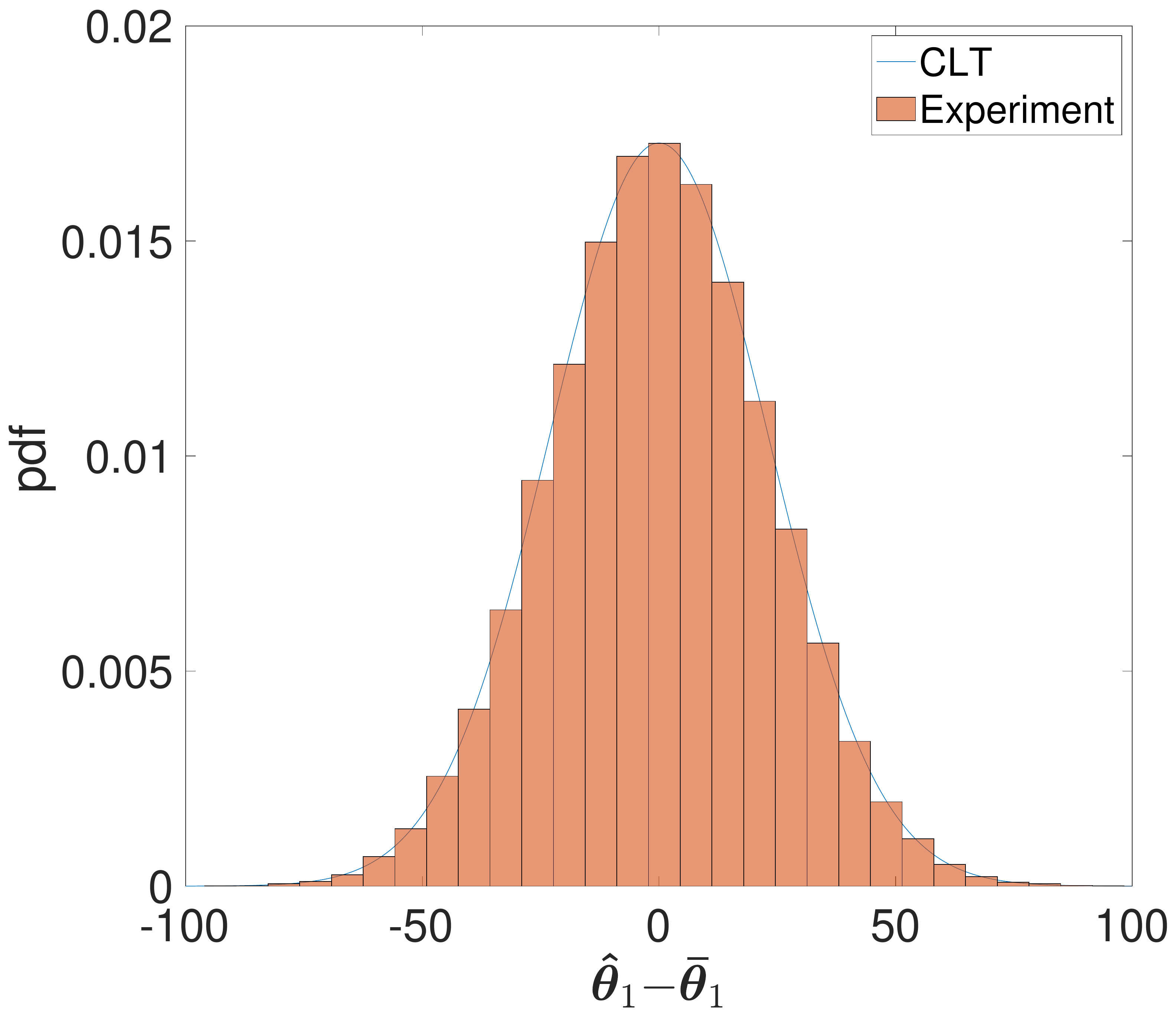}
}
\subfigure[Square]{
\includegraphics[width=0.46 \linewidth]{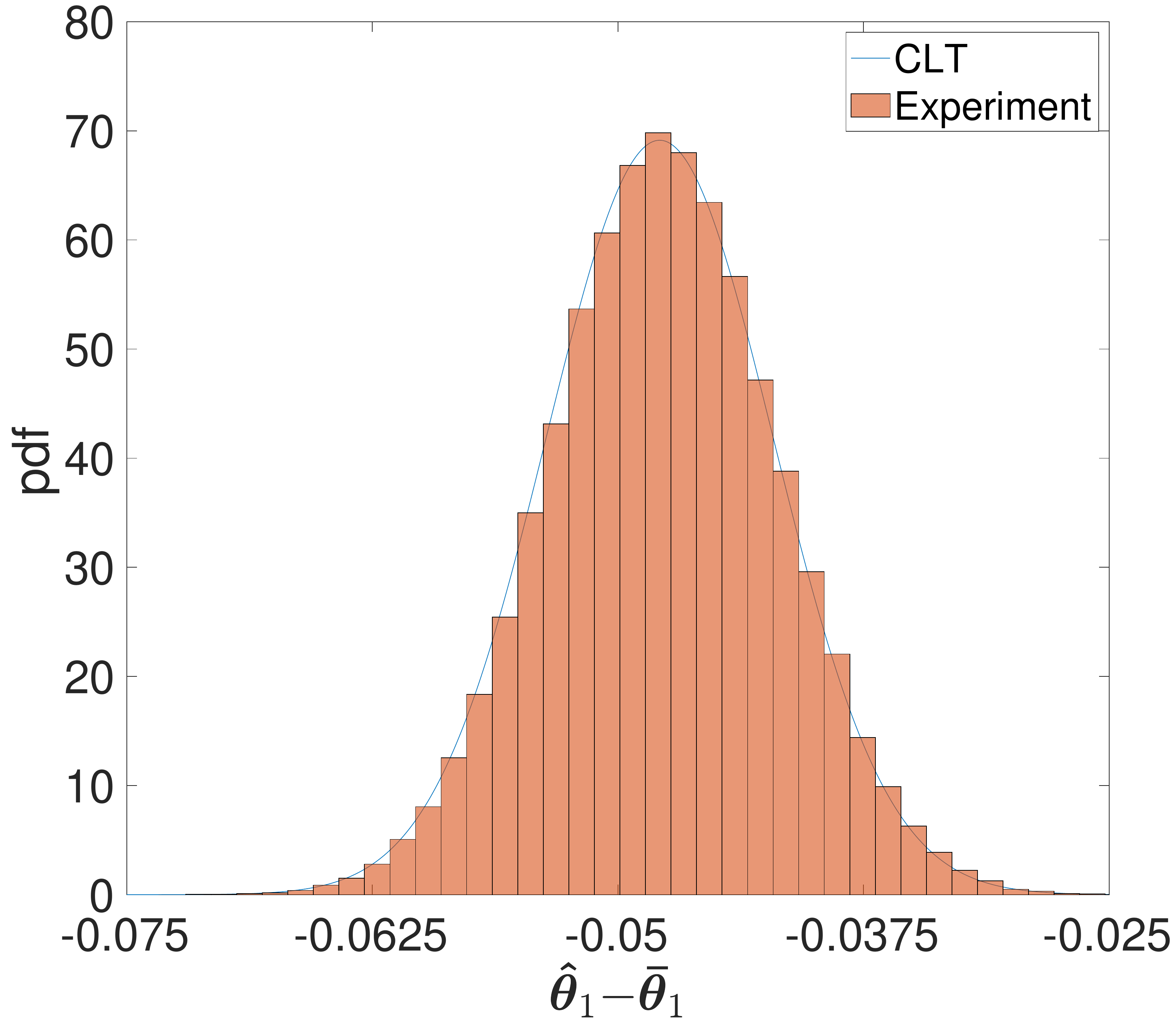}
}
\vspace{-0.1in}
\caption{Analysis vs. experimental results in our case study.}
\vspace{-0.2in}
\label{fig::bench}
\end{figure}

In the second set of experiments, we evaluate impacts of different LDP mechanisms, privacy budget and dimensionality on our re-calibration protocol. In particular, $\epsilon$ is varied in the set $\{0.1, 0.2, 0.4, 0.8, 1.6, 3.2\}$ for Laplace and Piecewise while in the set $\{0.1, 10, 100, 500, 1000, 5000\}$ for Square wave. We set a different range of $\epsilon$ for Square wave because its utility hardly varies with small $\epsilon$ \cite{li2020estimating}. To test the limit of our protocol, each user sends all dimensions of her perturbed tuple to the data collector. Accordingly, $\epsilon$ is partitioned according to respective dimensions. Figs. \ref{fig::NorBudCom}(a)-(c) plot MSE results with respect to $\epsilon$ under the Gaussian dataset, where users and dimensions are respectively set 100,000 and 100. Overall, both $L_1$- and $L_2$-regularization enhance the aggregation accuracy in all three LDP mechanisms. In contrast to $L_1$, the MSE of $L_2$ decreases at a slower rate as $\epsilon$ increases, mostly because the regularization weights of $L_2$ become so large under 100 dimensions that any change of $\epsilon$ has a minor impact on the regularized results. Although our protocol mostly increases the utilities, there exist some expectations. In Square wave, $L_2$ is outperformed by the current aggregation if $\epsilon=5,000$ while $L_1$ tends to get the same results as the current aggregation regardless of $\epsilon$. It is noteworthy, however, that MSE results of both $L_1$ and $L_2$ may become worse than the current aggregation. In essence, if the deviation of the LDP mechanism does not satisfy the threshold in either Lemma \ref{lem::L1para} or Lemma \ref{lem::L2para}, our re-calibration can be harmful. In this sense, regularization should not be heavily involved or even involved at all. Note that Square wave mechanism perturbs original data from $[0,1]$ to $[-b,b+1]$, where $b \to \frac{1}{2}$ if $\epsilon \to 0$ while $b \to 0$ if $\epsilon \to +\infty$ \cite{li2020estimating}. With so concentrated perturbation, its deviation in each dimension can be so small that the threshold of either regularization is not satisfied. That is why $L_2$ could make its utility even worse. In contrast, our protocol successfully enhances utilities of Laplace and Piecewise because their perturbations are rather large in high-dimensional space.

\begin{figure*}
\begin{center}
\subfigure[Laplace on Gaussian ($d=100$)]{
\includegraphics[width=0.31\linewidth]{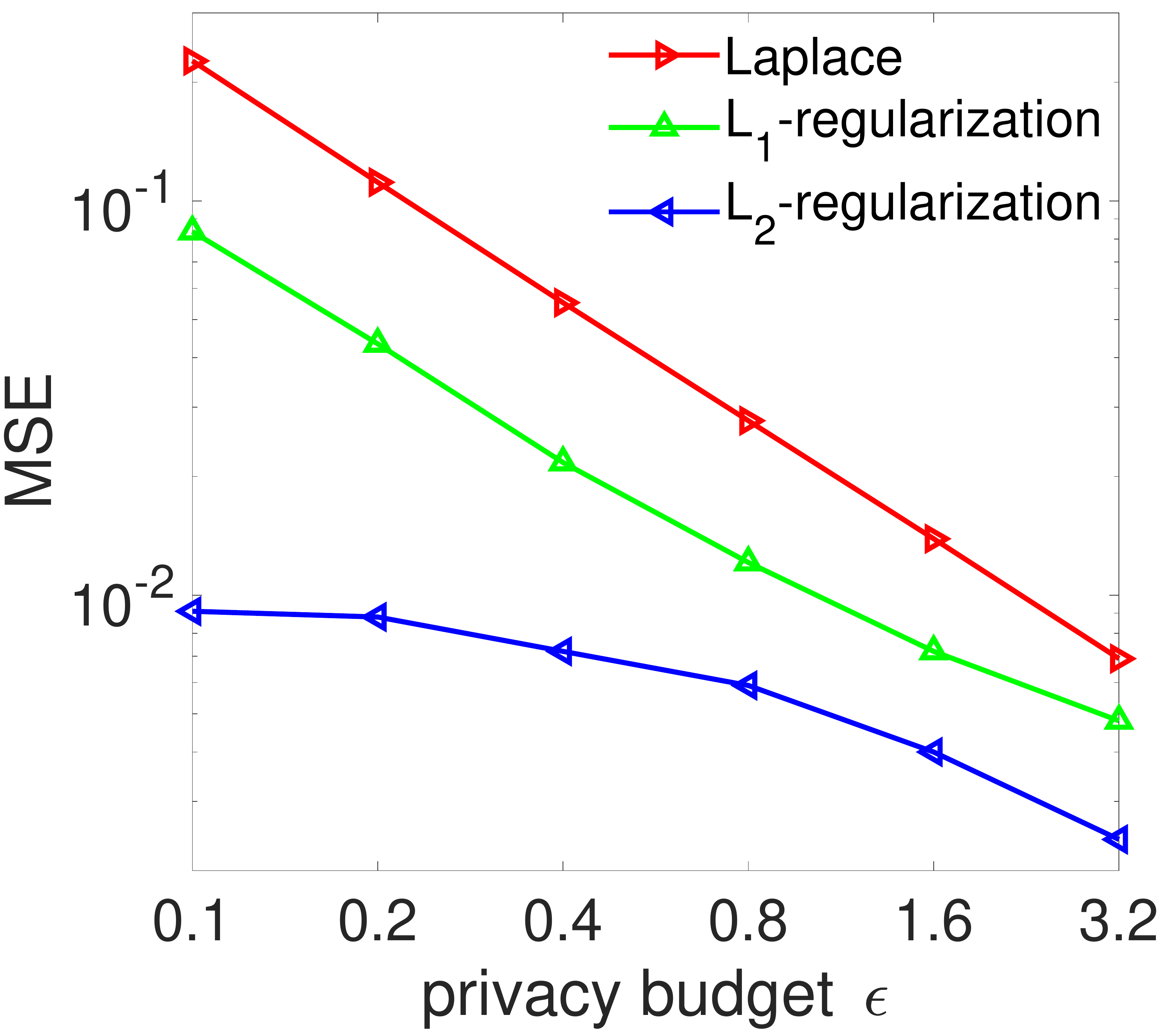}
}
\subfigure[Piecewise on Gaussian ($d=100$)]{
\includegraphics[width=0.31\linewidth]{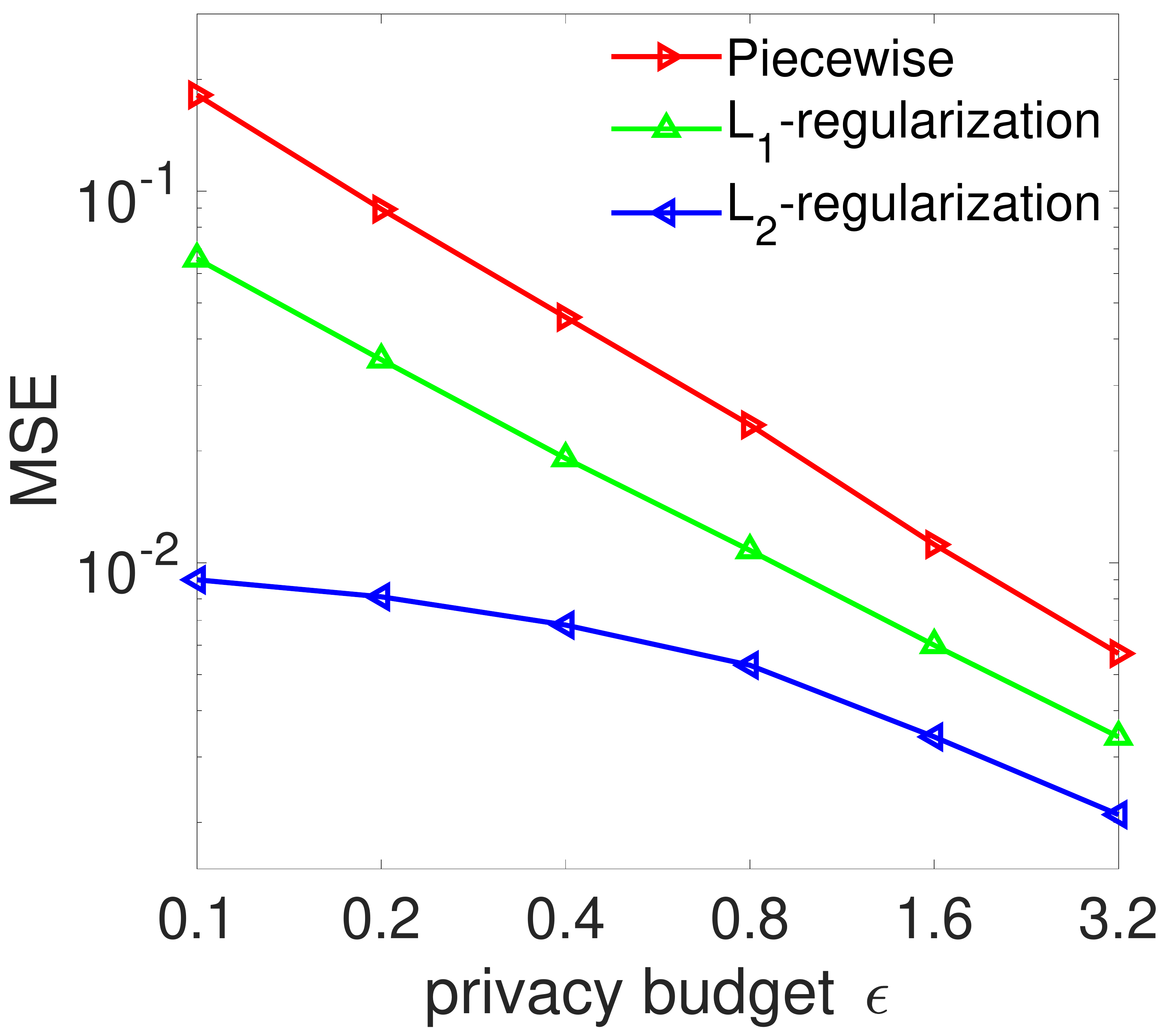}
}
\subfigure[Square on Gaussian ($d=100$)]{
\includegraphics[width=0.318\linewidth]{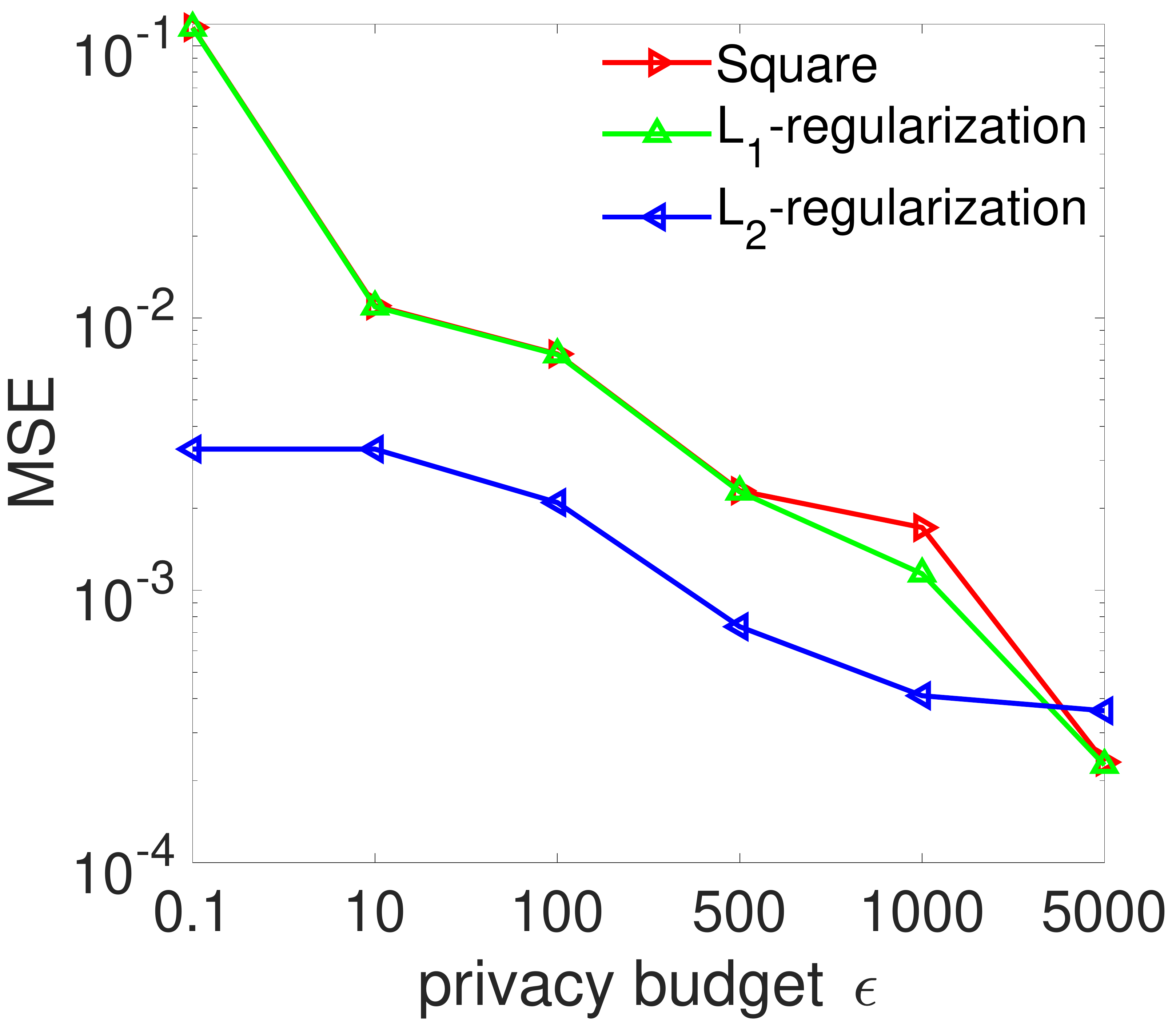}
}

\subfigure[Laplace on Poisson ($d=300$)]{
\includegraphics[width=0.312\linewidth]{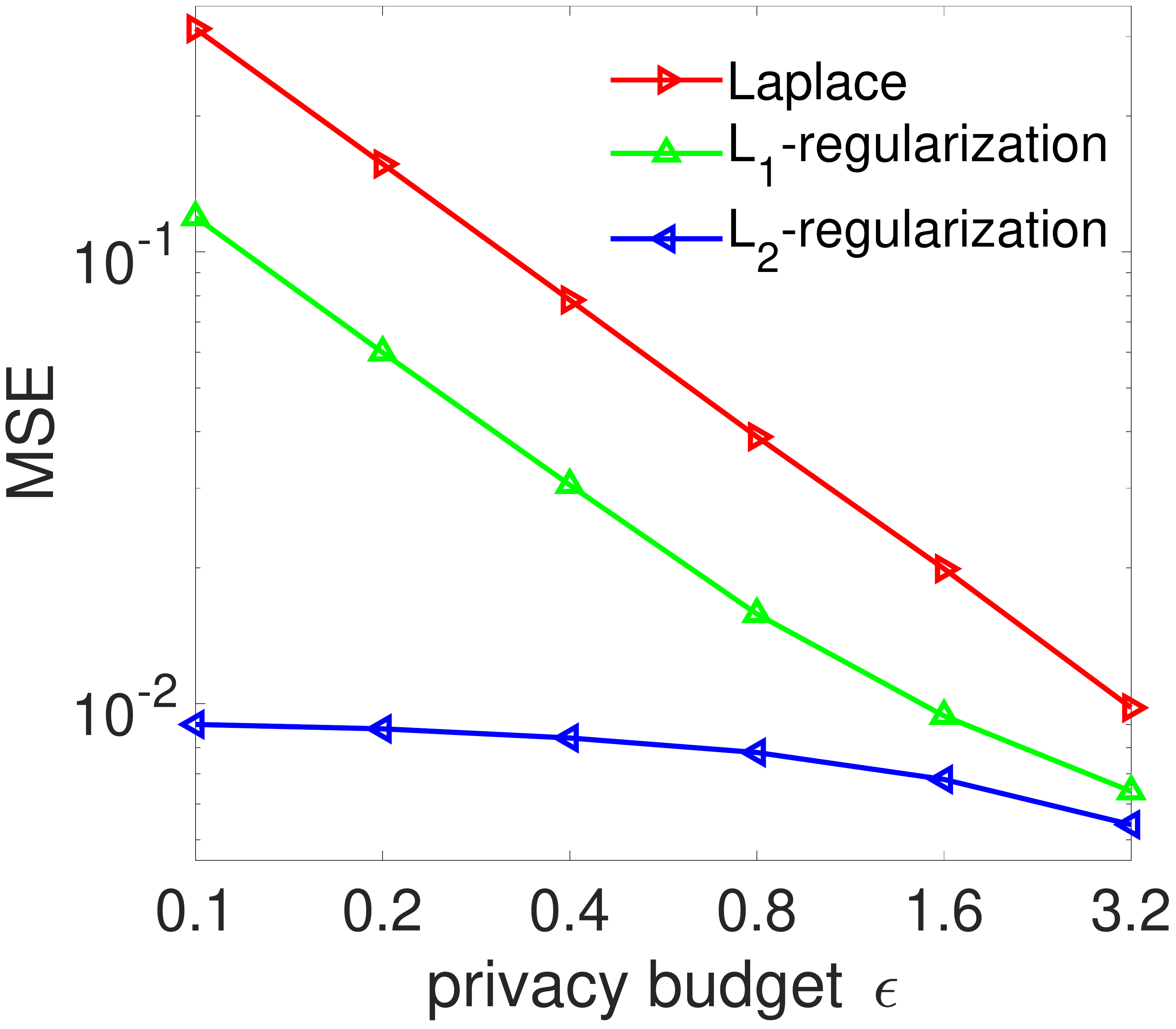}
}
\subfigure[Piecewise on Poisson ($d=300$)]{
\includegraphics[width=0.315\linewidth]{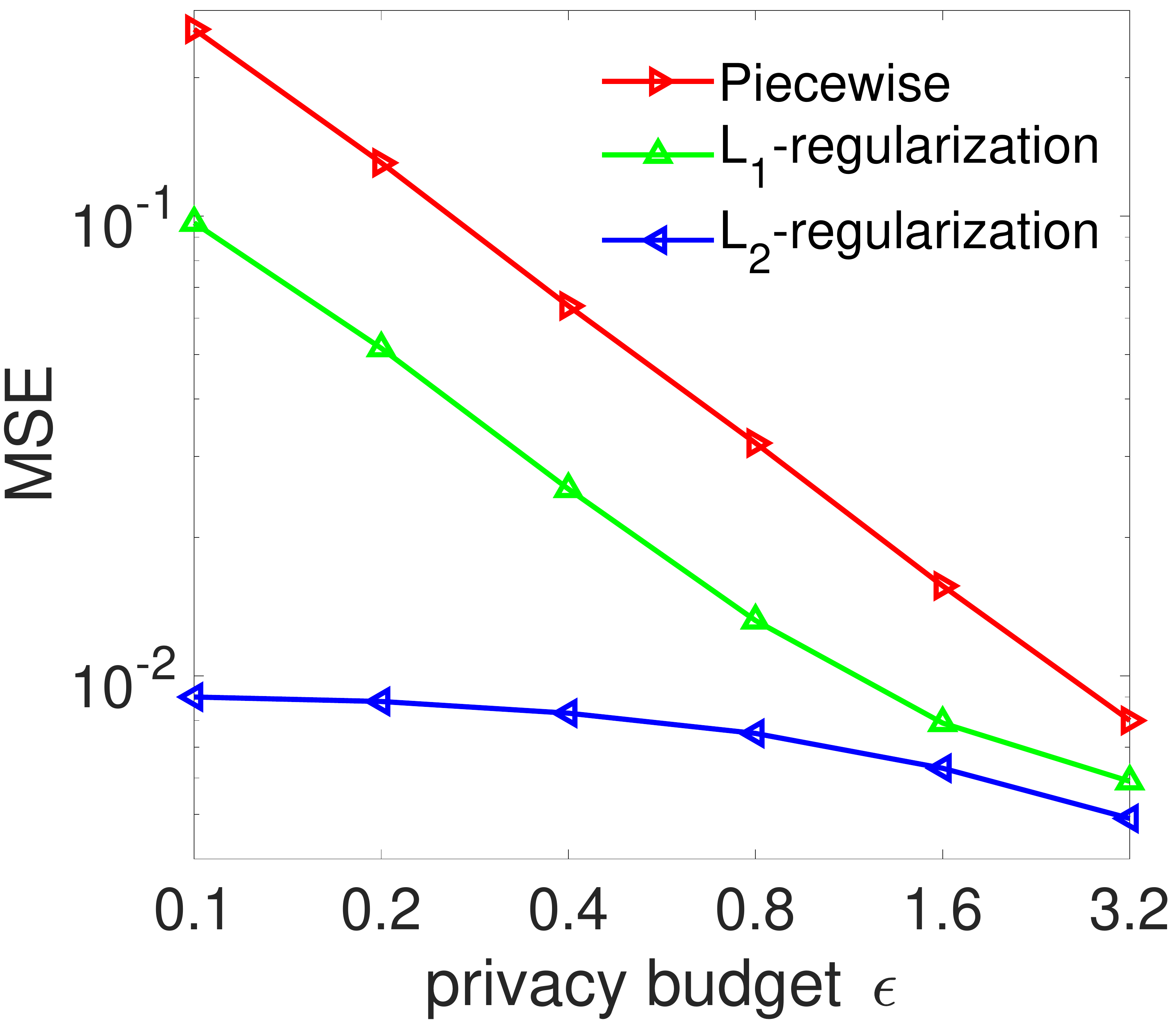}
}
\subfigure[Square on Poisson ($d=300$)]{
\includegraphics[width=0.317\linewidth]{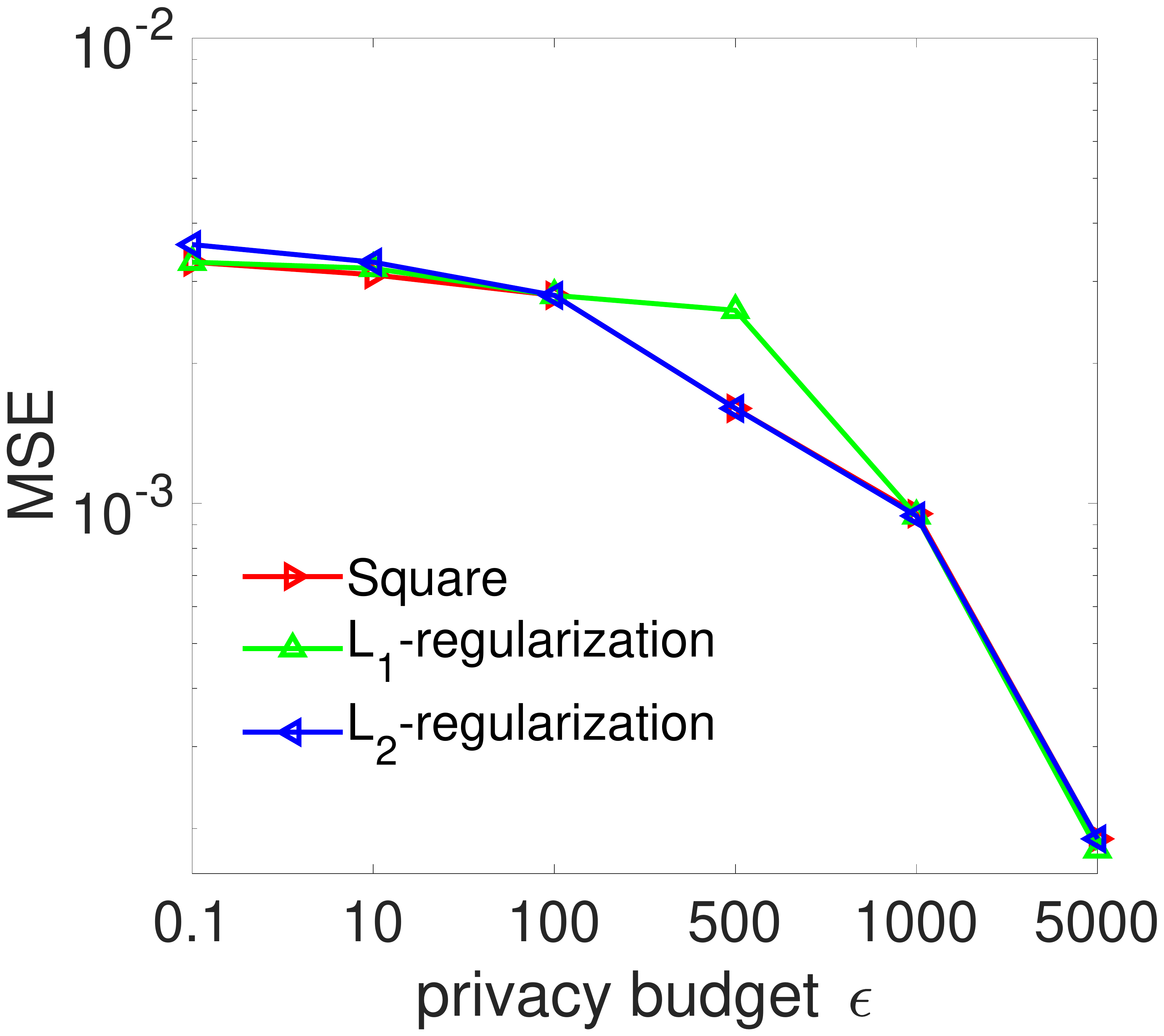}
}

\subfigure[Laplace on Uniform ($d=500$)]{
\includegraphics[width=0.315\linewidth]{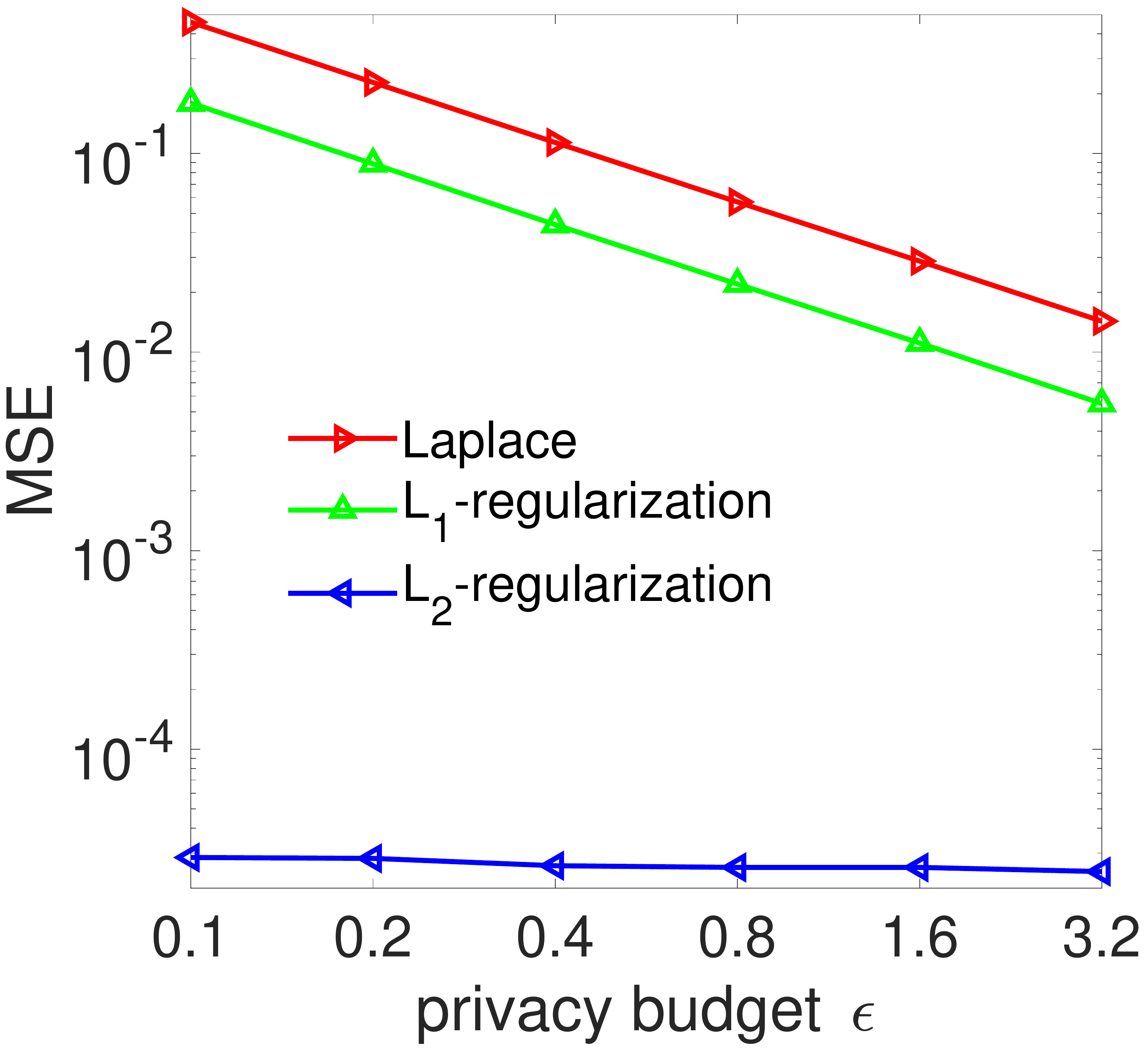}
}
\subfigure[Piecewise on Uniform ($d=500$)]{
\includegraphics[width=0.313\linewidth]{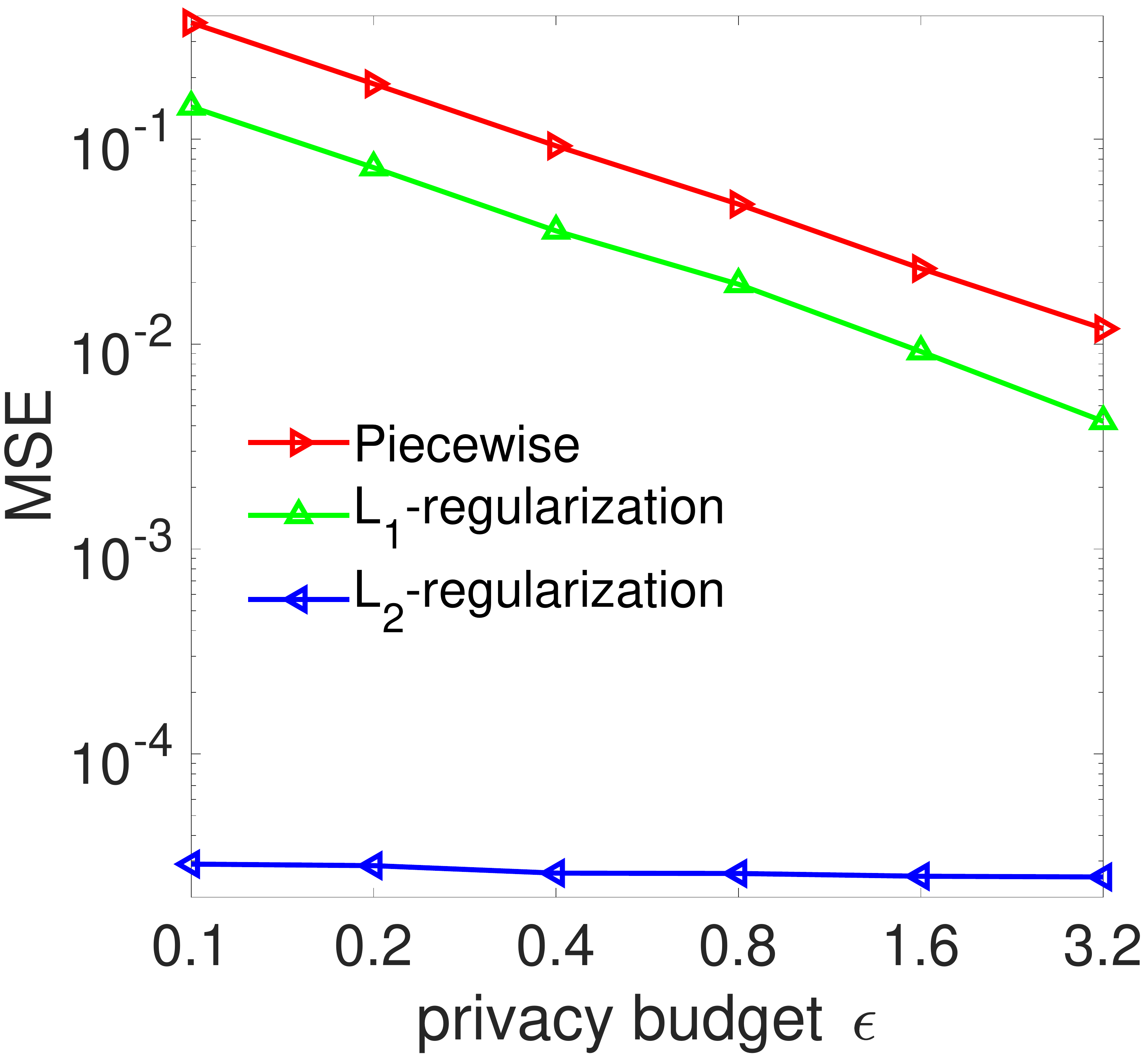}
}
\subfigure[Square on Uniform ($d=500$)]{
\includegraphics[width=0.317\linewidth]{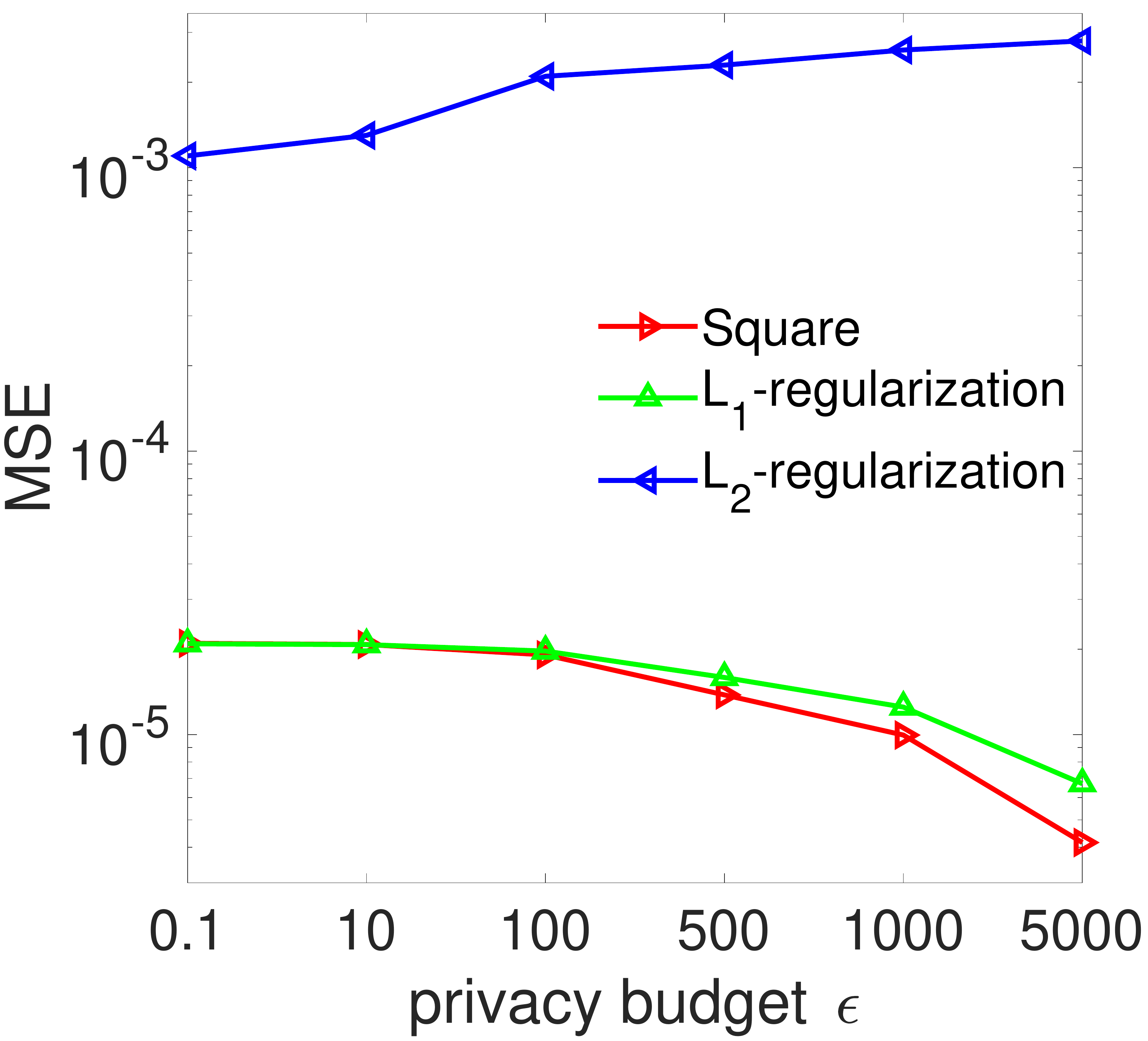}
}

\subfigure[Laplace on COV-19 ($d=750$)]{
\includegraphics[width=0.315\linewidth]{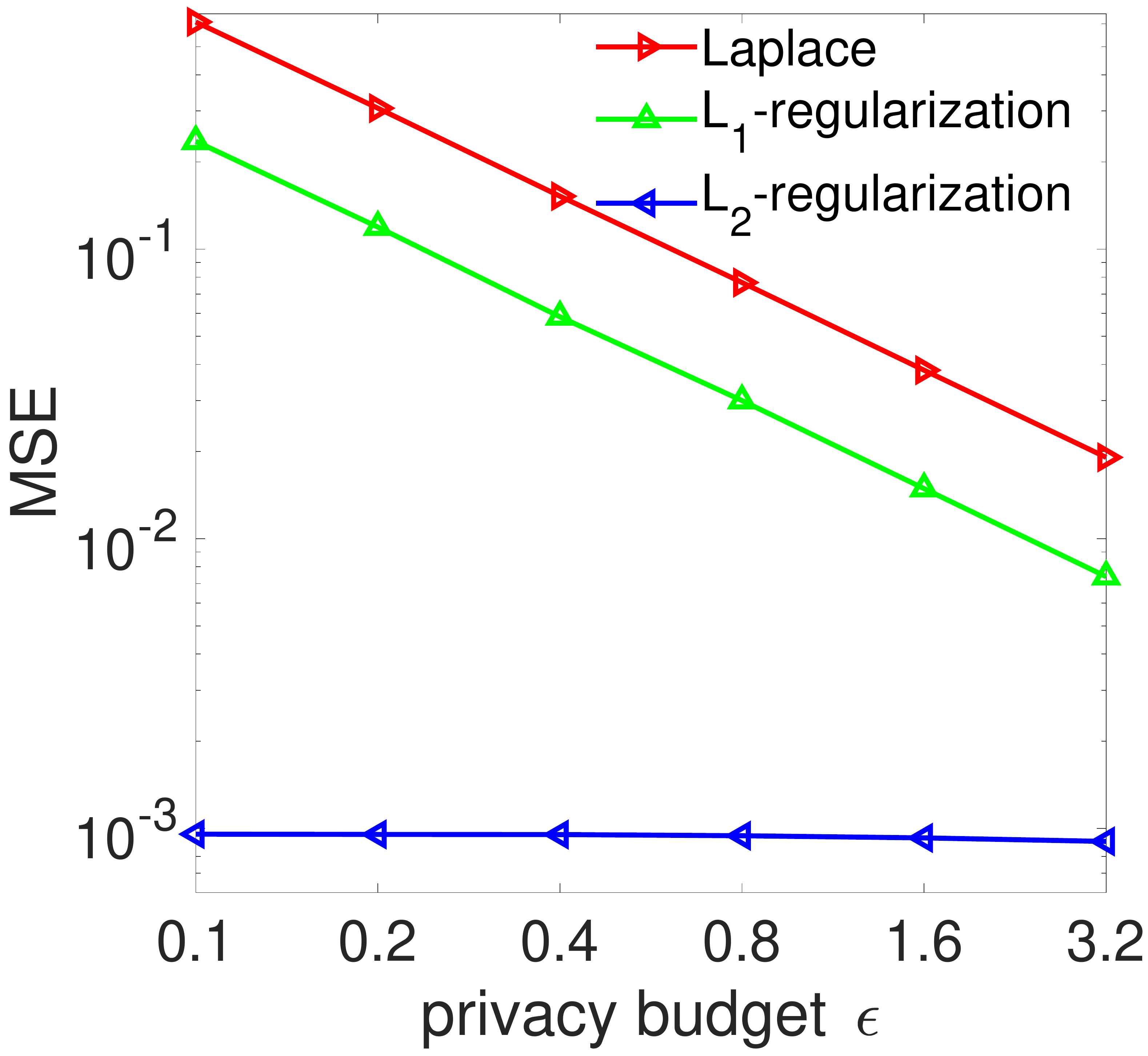}
}
\subfigure[Piecewise on COV-19 ($d=750$)]{
\includegraphics[width=0.3155\linewidth]{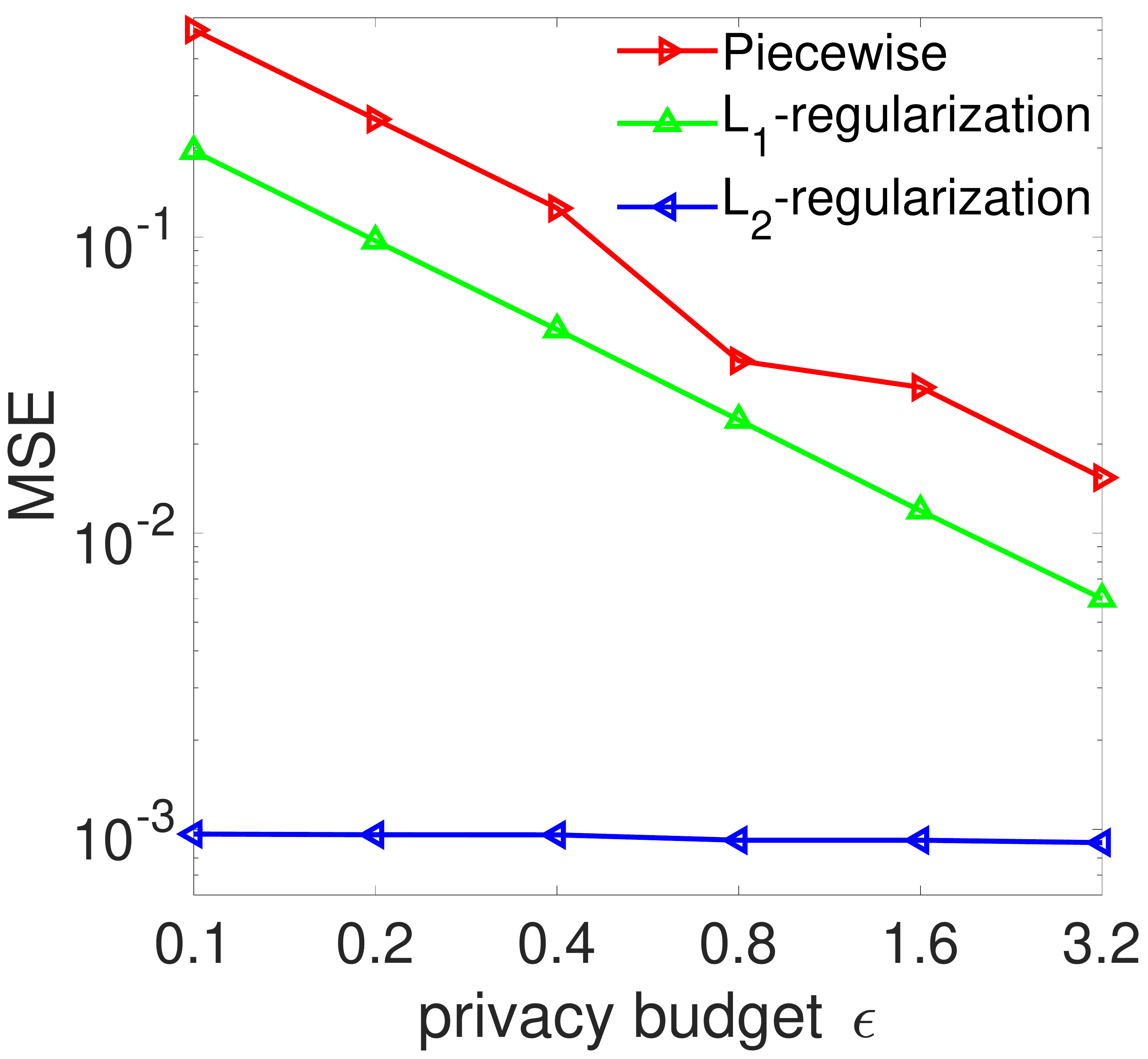}
}
\subfigure[Square on COV-19 ($d=750$)]{
\includegraphics[width=0.315\linewidth]{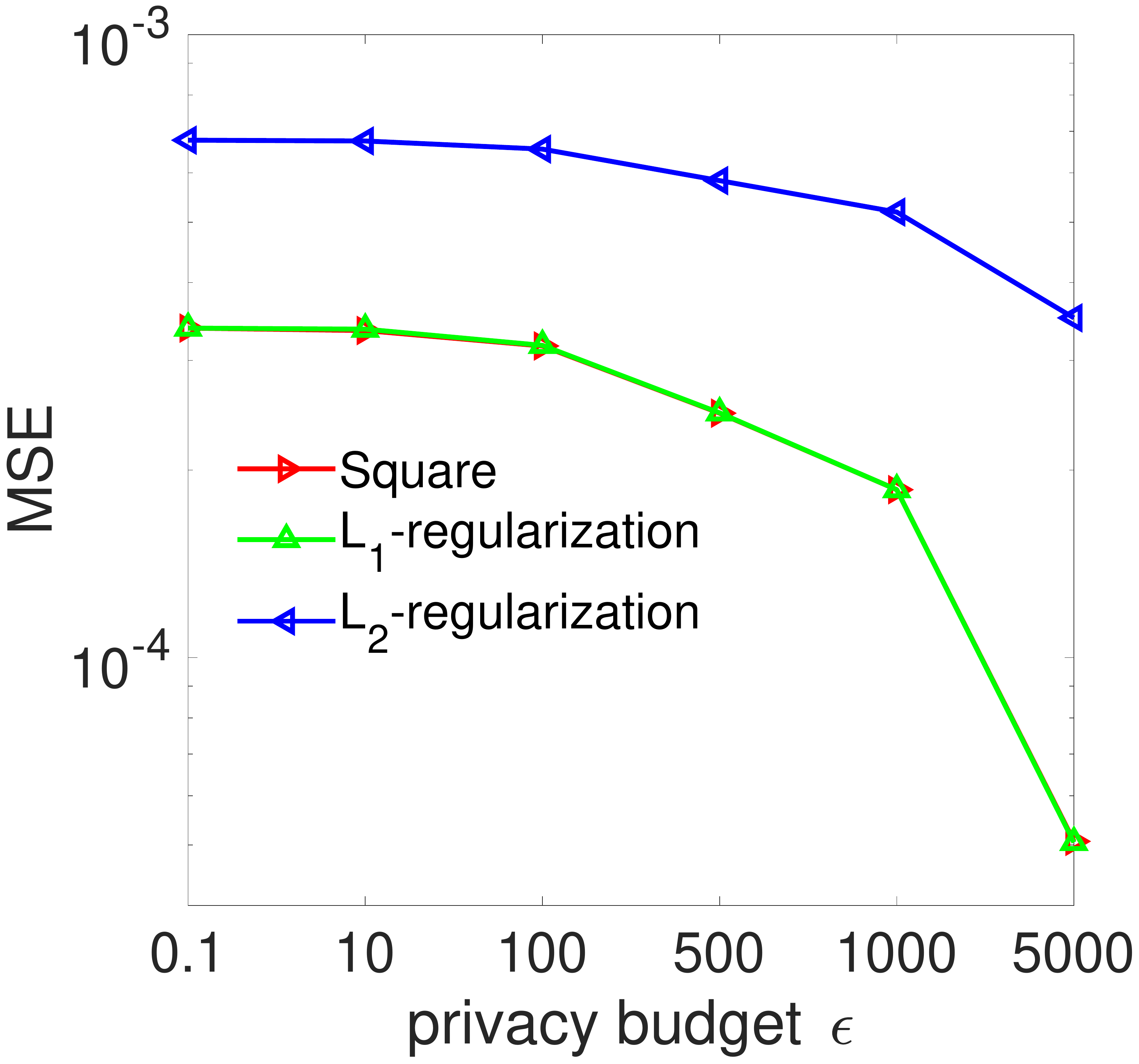}
}
\caption{MSE on various datasets and dimensions}
\label{fig::NorBudCom}
\end{center}
\end{figure*}
Next, we implement Poisson dataset (150,000 users, 300 dimensions), Uniform dataset (120,000 users, 500 dimensions) and COV-19 dataset (150,000 users, 750 dimensions) to repeat the utility enhancement experiment. $\epsilon$ is partitioned according to respective dimensions. Figs. \ref{fig::NorBudCom}(d)-(f), (g)-(i) and (j)-(l) show respective MSE results with regard to different datasets. In specific, Figs. \ref{fig::NorBudCom}(f), (i) and (l) confirm that our protocol is not suitable for Square wave whose deviation is already small enough in high-dimensional space while other figures indicate that both $L_1$ and $L_2$ enhance utilities of Laplace and Piecewise. Notably, MSE results of $L_2$ in Figs. \ref{fig::NorBudCom}(g), (h), (j) and (k) almost remain unchanged. Due to the extremely high dimensionality (e.g. d=500 and d=750), regularization weights of $L_2$ become so large that each entry of the enhanced mean is nearly zero. In this sense, MSE results of $L_2$ hardly change. 
\begin{figure}[H]
\centering
\subfigure[Laplace]{
\includegraphics[width=0.45\linewidth]{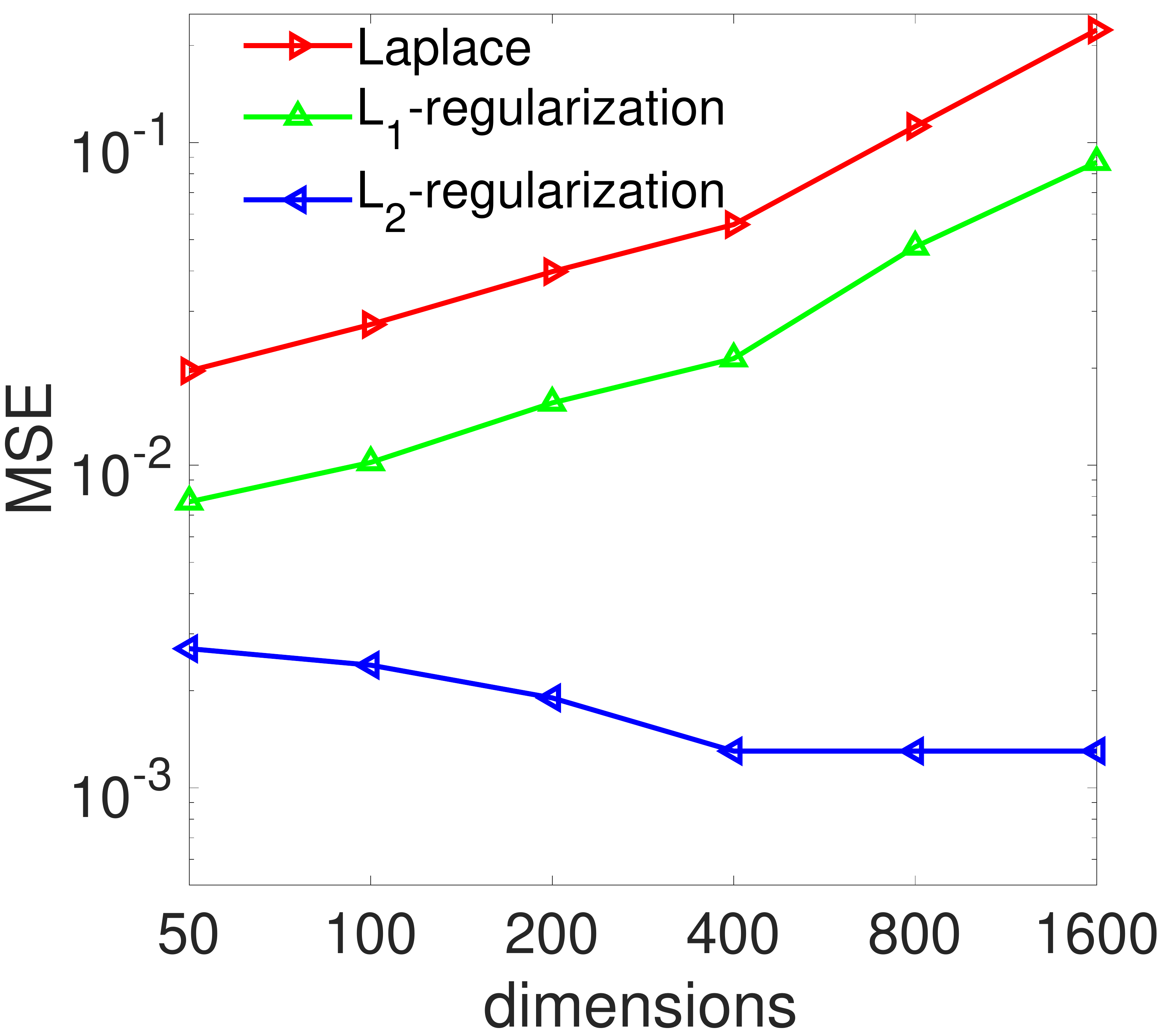}
}
\subfigure[Piecewise]{
\includegraphics[width=0.45\linewidth]{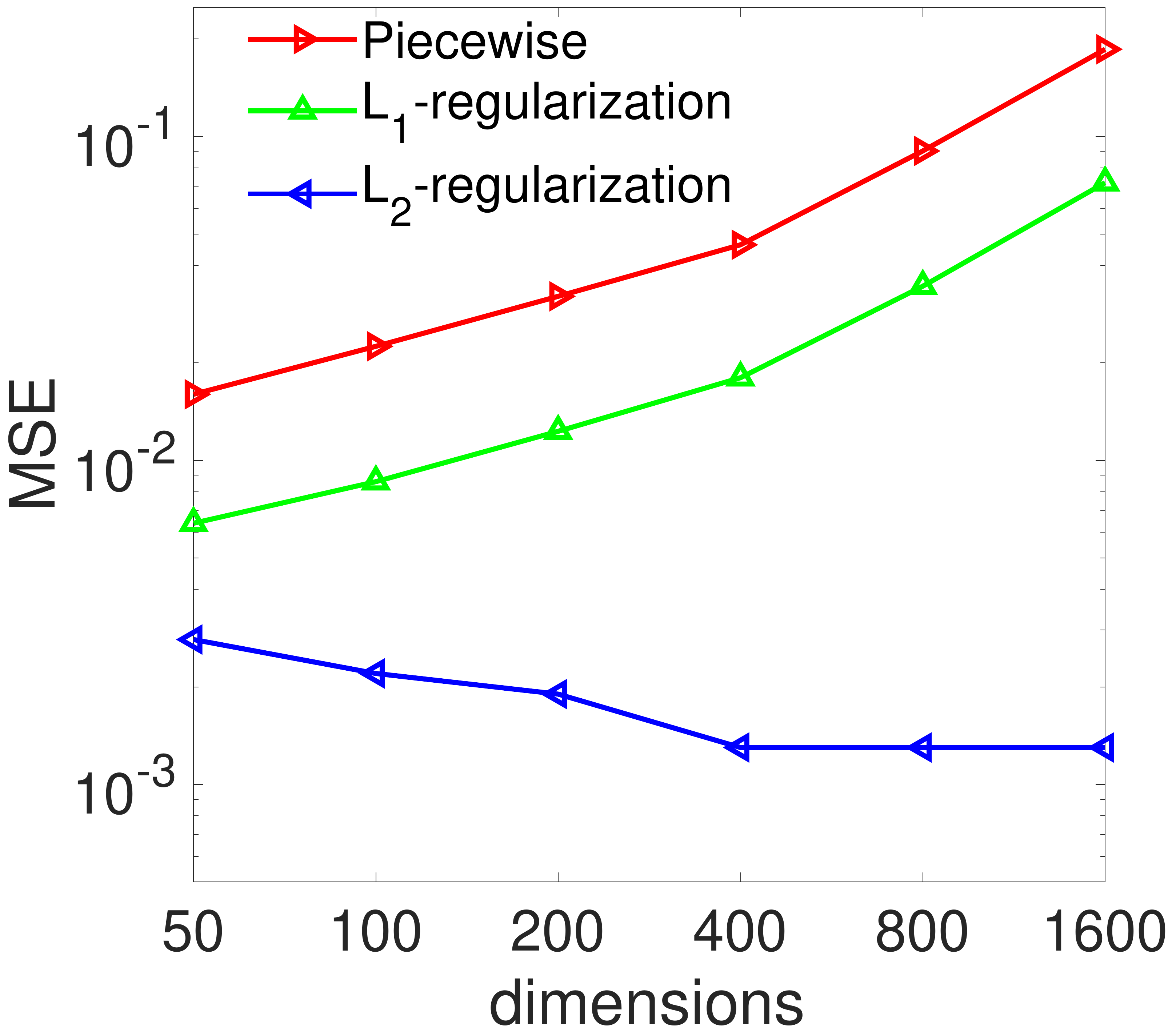}
}
\caption{MSE on COV-19 dataset with various dimensions.}
\label{fig::CovDimension}
\end{figure} 

In the third sets of experiments, we continue to evaluate the impact of dimensionality on our protocols under the COV-19 dataset, where $\epsilon$ is set 0.8, and the dimensionality varies in the set $\{50, 100, 200, 400, 800, 1600\}$. Since the dataset with dimensionality like 1600 is very hard to find, we randomly sample some dimensions from COV-19 dataset to make up. Fig. \ref{fig::CovDimension} shows MSE results of the Laplace and Piecewise, where our protocol enhances the current aggregation regardless of dimensionality. In particular, $L_2$ provides even better utilities as dimensionality increases, as opposed to both the current aggregation and $L_1$. The rationale is similar as above --- the regularization weights of $L_2$ are much larger than those of $L_1$ as dimensionality increases, which reduces the scale of perturbation more effectively. In this sense, MSE results of $L_2$ in both mechanisms decrease as dimensionality increases (e.g. $d=50$, $100$, $200$). As the dimensionality becomes extremely large (e.g. $d=400$, $800$, $1600$), regularization weights of $L_2$ become so large that each entry of enhanced mean is nearly zero. In this sense, MSE results of $L_2$ hardly change. 

	\section{Conclusion}
\label{sec::conclusion}
This work investigates utilities of mean estimation by LDP mechanisms in high-dimensional space. In terms of the deviation between the estimated mean and the true mean, we propose an analytical framework to evaluate any LDP mechanism. This framework provides closed-form evaluation on individual LDP mechanism. In addition, we propose \emph{HDR4ME} to re-calibrate the aggregation results from these LDP mechanisms to further enhance their utilities in high-dimensional space. Through theoretical analysis and extensive experiments, we confirm the generality and effectiveness of our analytical framework and re-calibration protocol under various datasets and parameter settings. 

For the future work, we plan to extend our work to other data type, e.g., set-value data, and more data analysis tasks, e.g., other statistics estimation and machine learning models.
	
	\section*{Acknowledgements}
\label{sec::ack}
This work was supported by National Natural Science Foundation of China (Grant No: 62072390 and 62102334), the Research Grants Council, Hong Kong SAR, China (Grant No: 15222118, 15218919, 15203120, 15226221 and 15225921), and Centre for Advances in Reliability and Safety (CAiRS) admitted under AiR@InnoHK Research Cluster.

\bibliographystyle{IEEEtran}
\bibliography{IEEEabrv,ref}
\end{document}